\documentclass[11.5pt]{article}
\usepackage{amsthm}
\usepackage{amsmath,bm}
\usepackage{enumerate}
\usepackage{amssymb}
\usepackage{latexsym}
\usepackage{amsfonts}
\usepackage{subfigure}
\usepackage{color}
\usepackage{secdot}
\usepackage{mathrsfs}
\usepackage{epsfig}
\usepackage{natbib}
\usepackage{appendix}
\newtheorem{theorem}{Theorem}[section]
\newtheorem{example}{Example}[section]
\newtheorem{counterexample}{Counterexample}[section]

\newtheorem{definition}{Definition}[section]

\newtheorem{lemma}{Lemma}[section]
\newtheorem{remark}{Remark}[section]

\begin{document}
	\title{Ordering results for extreme claim amounts based on random number of claims}
	\author{{\large { {\bf Sangita 
					Das}$^{a}$\thanks {Email address: sangitadas118@gmail.com,}}} \\
		{\small \it $^{a}$Theoretical Statistics and Mathematics Unit, Indian Statistical Institute, Bangalore-560059, India}\\\\
{ \em{\it To appear in Ricerche di Matematica.}}}
	\date{}
	\maketitle
	\begin{center}
		\noindent{\bf Abstract}
	\end{center}
    Consider two sequences of heterogeneous and independent portfolios of risks $T_1,T_2,\ldots$ and $T^*_{1}, T^*_{2},\ldots$ and, let $N_1$ and $N_2$ be two positive integer-valued random variables, independent of $T_i'$ and $T^*_i$, respectively. In this article, we investigate different stochastic inequalities involving $\min\{T_1,\ldots,T_{N_1}\}$ and $\min\{T^*_1,\ldots,T^*_{N_2}\},$ and $\max\{T_1,\ldots,T_{N_1}\}$ and $\max\{T^*_1,\ldots,T^*_{N_2}\}$ in the sense of usual stochastic order and reversed hazard rate order concerning maltivariate chain majorization order. These new results strengthen and generalize some of the well known results in the literature, including \cite{barmalzan2017ordering}, \cite{balakrishnan2018} and \cite{kundu2021_shock} for the case of random claim sizes. Different numerical examples  are provided to highlight the applicability of this work. Finally, some interesting applications of our results in reliability theory and auction theory are presented.
	
		\section{Introduction }
A semiparametric family of life distributions is a family that corresponds to a class of statistical models which provides a flexible alternative framework compared to fully parametric approaches. These models are extensively applied in insurance and actuarial science for modeling mortality, survival, and time-to-decrement processes (see \cite{marshall2007life} and \cite{marshall1997new}). A random variable $X$ is said to have a semiparametric family of distribution if its survival function can be written as
\begin{equation}\label{4}
   \bar{F}_{X}(x)=\bar{F}(\cdot;\alpha),~\alpha>0.
\end{equation}
Denote $\bar F(\cdot;\alpha)$, $f(\cdot;\alpha)$, $r(\cdot;\alpha)$ and $\tilde{r}(.;\alpha)$ are the survival, density, hazard rate and reversed hazard rate functions, respectively. Among all semiparametric models, the scale and proportional hazard models are one of the most widely recognized semiparametric models commonly used in reliability, survival studies and numerous other disciplines (see \cite{finkelstein2008failure, finkelstein2013stochastic,kumar1994,Cox1992}). Particularly in insurance analysis, scale models play a crucial role in assessing risks, predicting claims, and estimating policyholder behavior. Depending on different factors, such as coverage, region, or inflation, the size of claims can vary. To handle these situations, scale models are used to adjust these factors without changing the shape of the baseline distribution. Similarly, proportional hazard rate models have also been widely used to determine the timing of events such as mortality, illness, policy lapses, disability, and claims
while accounting for multiple risk factors. In particular, it is useful in life insurance, health
insurance, disability insurance, and policy retention studies because it allows for the inclusion
of various covariates (e.g., age, gender, lifestyle, economic conditions) without making strict
assumptions about the underlying baseline hazard function.

In an insurance period, for $i=1,\ldots,n$, let $U_{i}$ be a positive random variable represent the claim severity (or loss amount) for the $i$th risk, assuming that a claim occurs. Associated with $U_{i}$, let $J_{i}$ be a Bernoulli random variable with $E(J_{i})=p_i$, such that
\begin{eqnarray*}
	J_i= \displaystyle\left\{\begin{array}{ll} 1,
		& \textrm{ if the $i$th insured person makes random claim $U_{i}$},\\
		0, & \textrm{ otherwise},
	\end{array} \right.
\end{eqnarray*}
where $i=1,\ldots,n$. Then, the quantity $T_{i}=J_{i}U_{i}$ represent the individual claim amount corresponding to the $i$th insured person in a portfolio of risks $\{T_1,\ldots,T_n\}$. Denote $T_{1:n}=\min\{T_1,\ldots,T_{n}\}$ and $T_{n:n}=\max\{T_1,\ldots,T_{n}\}$ are the smallest and largest claim amounts, respectively, corresponding to the portfolio of risks $\{T_1,\ldots,T_n\}.$  Let us denote 
 $\psi:(0,1)\rightarrow (0,\infty)$ be a differentiable function. Also denote $\bm{\psi}(\bm{p})=(\psi(p_1),\ldots,\psi(p_n)).$ {In the context of risk analysis, the function $\psi(\cdot)$ represents the transformed vector of occurrence probabilities, will be used to modify probability distributions corresponding to risk assessment, premium calculation, reinsurance pricing, and capital allocation. Insurers use these transformations to change actual probabilities to display various risk attitudes, regulatory requirements, or pricing strategies. In particular, decreasing transformation will help the higher probability event to get less weight, and by the help of convex transformation, a risk-averse insurer increases the probability of large claims}.
 
In insurance analysis, it is of great importance to compare the amount of extreme claims to analyze the risk measure, as they provide important information to determine the annual lose or gain. Therefore, studying stochastic properties of extreme claim amounts are important both in practical and theoretical point of view.  
In this direction, a large amount of research have been done by many authors by considering heterogeneous independent/interdependent portfolios of risk having different general families of risk severities for the case of fixed claim size. Among them, \cite{barmalzan2017ordering} have considered independent heterogeneous general scale severities and established different ordering properties between smallest claim amounts with respect to usual stochastic order when the matrix of parameters $(\boldsymbol{\psi}(\boldsymbol{p}), \boldsymbol{\alpha};n)$ changes to $(\bm{\psi}(\bm{p}^*),\bm{\beta};n)$ in the sense of multivariate
chain majorization order and hazard rate order when the scale parameters are connected with weakly supermajorization order as well as weakly submajorization order. Consequently, \cite{balakrishnan2018} have investigated the ordering properties of the largest claim amounts arising from two sets of heterogeneous portfolios for independent observation according to the usual stochastic order when the matrix of parameters $(\boldsymbol{\psi}(\boldsymbol{p}), \boldsymbol{\alpha};n)$ changes to $(\bm{\psi}(\bm{p}^*),\bm{\beta};n)$ in the sense of row weakly majorization order. Subsequently, \cite{zhang2019ordering} have established sufficient conditions for comparing extreme claim amounts arising from two sets of heterogeneous insurance portfolios according to various stochastic orders. Recently \cite{Kundu2024} have investigated if the model parameters and the transform vectors are connected with weakly submajorization order then the largest claim amounts are comparable with respect to usual stochastic and reversed hazard rate orders for both independent/interdependent set-up. Different comparison results of extreme claim amounts having independent/interdependent
general claim severities have developed by several researchers, like location-scale severities (see, \cite{barmalzan2020}, \cite{Sameen2020} \cite{Das2021mcap} and \cite{sangita-balakrishnan2022}), transmuted-G model (see \cite{nadeb2020}), proportional odds model (see \cite{Panja2024}).

Randomness and variability are fundamental aspect of real-life data due to occurrence of the event, measurement errors, and differences among samples. 
 In many biological and agricultural studies, obtaining a fixed sample size can be challenging, as observations are often lost due to various factors, as a result the sample size becomes random. In insurance analysis, the number of claims received by an insurer from all its customers over the year is random. In actuarial science, randomness arises from different sources, closely related to uncertainty in future events. In finance and risk management, it is often interesting to get information about the
minimum and maximum loss from a portfolio of loans or securities when the number of assets or liabilities is random (see \cite{Laureano2008}). In transportation theory, to measure the accident-free distance of a shipment, such as explosives, the number of defective items may be random which will make a potential accident (see \cite{shaked1997}). In biostatistics, for the case of a cancer patient treatment the number of carcinogenic cells that are left to be activated after the first shot of the treatment is given will be random (see \cite{Cooner2007}). In hydrology, the study involving annual maximum rainfall or floods, it is very common that the number of storms or flood per year is random rather than fixed (see \cite{Koutsoyiannis2004}).

Let $X_{1},X_{2},\ldots$ be a sequence of independent random sample and $N$ be a positive integer-valued random variable, independent of $X_i$. Denote $X_{1:N}=\min\{X_{1},\ldots,X_{N}\}$ and $X_{N:N}=\max\{X_{1},\ldots,X_{N}\}$ are the random minima and random maxima, respectively. The random variable $X_{1:{N}}$ naturally emerges in transportation theory to represent the accident-free distance of a shipment, where $N$ defective items may detonate after traveling $X_1,\ldots, X_{N}$ miles, respectively, leading to a potential accident. In actuarial science, $X_{N:N}$ and $X_{1:N}$ represent the largest and smallest amount of claims in a particular time period. Thus, comparing two random maxima or minima stochastically is of significant interest to researchers owing to its important applications across various areas of statistics. 
In this context, relatively little work has been done by researchers on general families of distributions. Among them, in \cite{shaked1997} the authors have considered two different independent and identically distributed samples $X_1,X_2,\ldots$ and $Y_1,Y_2,\ldots$ having the same (random) sample size ($N$) and established that if $X_{i}\leq_{st}Y_{i},$ for $i=1,\ldots,N$ then $X_{1:{N}}\leq_{st}Y_{1:{N}}$ and $X_{N:{N}}\leq_{st}Y_{N:{N}}.$ Also, they have proved that 
if $N_1$ and $N_2,$ are connected with Laplace transform and Laplace transform ratio orders  then the usual 
stochastic, hazard rate, reversed hazard rate and likelihood ratio orders are hold between 
$X_{1:{N_1}}$, $X_{1:{N_2}}$ and $X_{N_{1}:{N_1}}$, $X_{N_{2}:{N_2}}$, where $N_1$ and $N_2$ are two positive integer-valued random variables, independent of $X_{i}'$s. Following this, \cite{bartoszewicz2001} have showed that if convex, star and super-additive orders hold between $X_1$ and $Y_1$, then $X_{1:N}$ and $X_{N:N}$ are smaller than $Y_{1:N}$ and $Y_{N:N},$ respectively. After that, \cite{Li2004} have established  some certain conditions such that the right spread and the increasing convex orderings hold between $X_{N:N}$ and $Y_{N:N}$. Furthermore, they have showed that total time on test transform 
and increasing concave orderings hold between $X_{1:N}$ and $Y_{1:N}$. Subsequently, \cite{Ahmad2007} have reported that reversed preservation property of right spread and total time on test transform orders under random minima and maxima. Very recently, in \cite{chowdhury2024} the authors have proved several comparison results of random minima and maxima from a random number of non-identical random variables. To get an overview of the results on stochastic comparisons of random maxima and minima, the interested reader may refer to \cite{Nanda2008}, \cite{Manesh2023}, \cite{Sangita_bala_26}.

Motivated from these work in literature, let us consider $\{U_{1},\ldots,U_{N_1}\}$
and $\{V_{1},\ldots,V_{N_2}\}$ be two sets of independent and heterogeneous random variables such that $i=1,\ldots,{N_1},$ $U_{i}\sim
\bar{F}(x;\alpha_{i})$
and for $i=1,\ldots,{N_2},$ $V_{i}\sim
\bar{F}(x;\beta_{i}).$ Also, let $T_{i}=J_{i}U_{i}$ and
$T_{i}^{*}=J_{i}^{*}V_{i}$, where $J_{i}$ and $
J_{i}^{*}$ be independent Bernoulli random variables, independent of $U_{i}'$s and $V_{i}'$s such that
$E(J_{i})=p_{i}$ and $E(J_{i}^{*})=p_{i}^{*}$ and, $N_1$ and $N_2$ be another two positive integer-valued random variables independently of $T_{i}'$s and ${T^{*}_{i}}'$s, respectively. Under the above set-up, the smallest and largest claim amount arising from
$T_{1},\ldots,T_{N_1}$ and $T_{1}^{*},\ldots,T_{N_2}^{*}$
are, respectively, denoted by
$T_{1:{N_1}}=\min\{T_{1},\ldots,T_{N_1}\},$
$T_{1:N_2}^{*}=\min\{T_{1}^{*},\ldots,T_{N_2}^{*}\},$
$T_{{N_1}:{N_1}}=\max\{T_{1},\ldots,T_{N_1}\}$ and
$T_{{N_2}:{N_2}}^{*}=\max\{T_{1}^{*},\ldots,T_{N_2}^{*}\}$, respectively.

To best of our knowledge, the proposed model presented in this study has not been investigated in the existing literature. In this article, our aim is to study the stochastic properties of random claims $(T_{1:N_{1}} \text{~and~} T^{*}_{1:N_{2}})$ and $(T_{N_{1}:N_{1}} \text{~and~} T^{*}_{N_{2}:N_{2}}),$ respectively, when the matrix of parameters $(\boldsymbol{\psi}(\boldsymbol{p}), \boldsymbol{\alpha};n)$ changes to $(\bm{\psi}(\bm{p}^*),\bm{\beta};n)$ in the sense of the row weak
majorization order in the space $M_{n}$ based on the usual stochastic and reversed hazard rate orders under the condition that ``$\psi$ is  strictly decreasing convex function''.  To prove our main result based on usual stochastic order presented in Theorem \ref{th_w_st_small}, we first establish the results (see Theorems \ref{th_m_st_alpha_small} and \ref{th_m_st_psi_small}) based on vector chain mejorization order for fix sample size and then using the condition $N_1\leq_{st}N_2$ we prove $T_{{1}:{N_1}}$  is stochastically less than $T_{{1}:{N_2}}.$ Similar idea is used for Theorem \ref{th_w_st}, which compares $T_{{N_1}:{N_1}}$ and $T_{{N_2}:{N_2}}$ stochastically. To prove the results based on reversed hazard rate order for same and random sample size $N$ (see Theorems \ref{th_m_rh_lrg} and \ref{th_m_rh_sm}), we apply some Theorems from \cite{chowdhury2024}, and using the idea of Theorem \ref{th_w_st_small}. As a consequence, our results extend the results in \cite{barmalzan2017ordering}, \cite{balakrishnan2018} and  \cite{kundu2021_shock} to a general set-up. Moreover, these results established here provide important insight into determining the final price in auction theory and the best reliability system in reliability theory.

The remaining part of the article is organized as follows. In Section \ref{p}, some important definitions and preliminary results are presented. The ordering results pertaining to different stochastic orders, such as the usual stochastic and reversed hazard rate orders, are developed in Section \ref{s1} using the concepts of matrix row majorization and row  majorization orders. In Section \ref{app}, we consider some applications of our established theoretical results for the purposes of illustration in reliability theory and auction theory. Finally, Section \ref{conclusion} provides a conclusion of this work.
\section{Preliminary}\label{p}

In this section, we review some important definitions and well-known concepts involving the notion of majorization and stochastic orders. The reader may consult to \cite{shaked2007stochastic} and \cite{Marshall2011} for further details on these notions.
Let $V_1$ and $V_2$ be two univariate random variables having
 density functions (PDFs) $f_{V_1}(\cdot)$ and $f_{V_2}(\cdot)$, distribution functions (CDFs) $F_{V_1}(\cdot)$ and $F_{V_2}(\cdot)$, survival functions $\bar
F_{V_1}(\cdot)$ and $\bar F_{V_2}(\cdot)$, and reversed hazard rate functions $\tilde{r}_{V_1}(\cdot)=f_{V_1}(\cdot)/
 F_{V_1}(\cdot)$,  $\tilde r_{V_2}(\cdot)=f_{V_2}(\cdot)/
F_{V_2}(\cdot)$, respectively. 
\begin{definition}
	A random variable $V_1$ is said to be smaller than $V_2$ in the
	\begin{itemize}
		\item reversed hazard rate order (denoted by $V_1\leq_{rh}V_2$)
		if $\tilde r_{V_1}(x)\leq \tilde r_{V_2}(x)$, for all $x$;
		\item usual stochastic order (denoted by $V_1\leq_{st}V_2$) if
		$\bar F_{V_1}(x)\leq\bar F_{V_2}(x)$, for all $x.$
	\end{itemize}
\end{definition}
Let $\boldsymbol{a} =
\left(a_{1},\cdots,a_{n}\right)$ and $\boldsymbol{b} =
\left(b_{1},\ldots,b_{n}\right)$ be two $n$-dimensional vectors, where $\boldsymbol{a}~,\boldsymbol{b}\in\mathbb{A}$. Here, $\mathbb{A} \subset \mathbb{R}^{n}$ and $\mathbb{R}^{n}$ is
a $n$-dimensional Euclidean space. Furthermore, let the order coordinates of the vectors $\boldsymbol{a}$ and $\boldsymbol{b}$ be $a_{1:n}\leq \cdots \leq a_{n:n}$ and $b_{1:n}\leq\cdots \leq b_{n:n},$ respectively. Denote $\mathcal{J}_{n}=\{1,2,\ldots,n\}.$

\begin{definition}\label{definition2.2}
	A vector $\boldsymbol{b}$ is said to 
	\begin{itemize}
		\item  majorize the vector $\boldsymbol{a},$ (denoted
		by $\boldsymbol{a}\preceq^{m} \boldsymbol{b}$), for $l\in \mathcal{J}_{n-1},$ $\sum_{i=1}^{l}a_{i:n}\geq \sum_{i=1}^{l}b_{i:n}$ and
		$\sum_{i=1}^{n}a_{i:n}=\sum_{i=1}^{n}b_{i:n};$
			\item weakly supermajorize vector $\boldsymbol{b},$ denoted
		by $\boldsymbol{a}\preceq^{w} \boldsymbol{b}$, if
		$\sum_{i=1}^{l}a_{i:n}\geq \sum_{i=1}^{l}b_{i:n}$ for $l\in\mathcal{J}_n.$				
	\end{itemize}
\end{definition} 

In the following, we present the definition of Schur-convex and
Schur-concave functions and, an important lemma which will be useful in the consequence section.
\begin{definition}
	A function $h:\mathbb{R}^n\rightarrow \mathbb{R}$ is said
	to be Schur-convex (Schur-concave) on $\mathbb{R}^n$ if \\
	$$\boldsymbol {x}\overset{m}{\succeq}\boldsymbol{ y}\Rightarrow
	h(\boldsymbol { x})\geq( \leq)h(\boldsymbol { y}) \text{, for all } \boldsymbol { x}, \boldsymbol
	{ y} \in \mathbb{R}^n.$$
\end{definition}
\begin{lemma}(\cite{Marshall2011})\label{lem2.1}
Consider the real-valued continuously differentiable function $\phi$ on $J^n,$ where $J\subseteq\mathcal{R}$ is an open interval. Then, $\phi$ is Schur-convex $(\text{Schur-concave})$ on $J^n$ if and only if $\phi$ is symmetric on $J^n$, and for all $i\neq j$ and all $\boldsymbol{u}\in J^n,$
$$(u_i-u_j)\left(\frac{\partial\phi(\boldsymbol{u})}{\partial u_i}-\frac{\partial\phi(\boldsymbol{u})}{\partial u_j}\right)\geq(\leq )0,$$
 where $\frac{\partial\phi(\boldsymbol{u})}{\partial u_i}$
 denotes the partial derivative of $\phi(\bm{u})$ with respect to its $i$-th argument.
\end{lemma}
Next, we discuss the concept of majorization on
matrices. A square matrix $\Pi$ is said to be a permutation matrix 
if each row and column
have exactly one unity and zeros elsewhere.
Such matrices can be constructed
by interchanging rows (or columns) of the
$n\times n$ identity matrix $I_{n}$. Therefore, a
$T$-transform matrix is of the form
\begin{eqnarray}
T_{w}=wI_{n}+(1-w)\Pi,~0\le w\le 1.
\end{eqnarray}

Consider two $T$-transform matrices
$T_{w_{1}}=w_{1}I_{n}+(1-w_{1})\Pi_{1}$ and
$T_{w_{2}}=w_{2}I_{n}+(1-w_{2})\Pi_{2},$ where
$\Pi_{1}$ and $\Pi_{2}$ are two permutation
matrices and $0\leq w_1, w_2\leq 1$. If
$\Pi_{1}=\Pi_{2},$ then the matrices $T_{w_{1}}$ and
$T_{w_{2}}$ have the same structure, otherwise they have different
structures. In the following, we describe the concept of
multivariate majorization.

\begin{definition}\label{def2.4}
Consider two $m\times n$ matrices $A=\{a_{ij}\}$ and
$B=\{b_{ij}\}$, where $i=1,\ldots,m$ and $j=1,\ldots,n$. Then, 
\begin{itemize}
    
\item $A$ is said to majorize $B,$
(denoted by $A>>B$), if there exists a finite set of
$n\times n$ $T$-transform matrices $T_{w_{1}},\ldots,T_{w_{k}}$ such
that $B=A T_{w_{1}}\ldots T_{w_{k}};$
\item $A$ is said to chain majorize $B,$
(denoted by $A>B$), if $B=A P,$ where the $n\times n$ matrix $P$ is doubly stochastic. Since a product of $T$-transform is double stochastic, it follows that $A>>B\Rightarrow A>B;$
\item $A$ is said to row majorize $B,$ (denoted by $A\overset{row}{>}B$), if $\boldsymbol{a}^{R}_{i}\overset{m}{\succeq}\boldsymbol{b}^{R}_{i}$ for $i=1,\ldots,m,$ where $\boldsymbol{a}_{1}^{R},\ldots,\boldsymbol{a}_{m}^{R}$ and $\boldsymbol{b}_{1}^{R},\ldots,\boldsymbol{b}_{m}^{R}$ are the rows of $A$ and $B,$ so that these quantities are row vectors of length $n.$ It is known that $A>B\Rightarrow A\overset{row}{>}B;$ 
\item $A$ is said to row weakly majorize $B,$ (denoted by $A\overset{w}{>}B$), if $\boldsymbol{a}^{R}_{i}\overset{w}{\succeq}\boldsymbol{b}^{R}_{i}$ for $i=1,\ldots,m,$ where $\boldsymbol{a}_{1}^{R},\ldots,\boldsymbol{a}_{m}^{R}$ and $\boldsymbol{b}_{1}^{R},\ldots,\boldsymbol{b}_{m}^{R}$ be the rows of $A$ and $B,$ so that these quantities are row vectors of length $n.$ It is known that $A\overset{row}{>}B\Rightarrow A\overset{w}{>}B.$ 
\end{itemize}
\end{definition}
For detailed discussion on this topic, we may refer to \cite{Marshall2011}. 
For simplicity, from now on we denote the matrix
of order $m\times n$ as
$(\bm{r_{1}},\bm{r_{2}},\ldots,\bm{r_{m}};n),$
where the vectors with real value
$\bm{r_{1}},\bm{r_{2}},$ $\ldots,$ $\bm{r_{m}}$
are the first, second, $\ldots,$ $m$-th rows,
respectively, each having $n$ elements. Let us
denote
\begin{eqnarray*}
    M_{n} &=& \left\lbrace (\bm{x},\bm{y};n) :\ x_{i} > 0, ~y_{j} > 0 \text{ and }
(x_{i} - x_{j})(y_{i} - y_{j}) \geq 0,~ i, j=1,\ldots,n \right\rbrace
\end{eqnarray*}
and $h'( z)=\frac{d h(z)}{d z}.$

\section{Ordering results}\label{s1}
In this section, we consider two sets of independent and heterogeneous insurance portfolios of risks $\{T_{1},\ldots,T_{N_1}\}$ and $\{T^{*}_{1},\ldots,T^{*}_{N_2}\}$ with $T_{i}=J_{i}U_{i}$ and $T^{*}_{i}=J^{*}_{i}V_{i}$ where for $i=1,\ldots,N_1,$ $U_{i}\sim \bar{F}(\cdot;\alpha_{i})$ and $J_{i}\sim Bernouli(p_{i})$ independent of $U^{'}_{i}$s and for $i=1,\ldots,N_2,$ $V_{i}\sim \bar{F}(\cdot;\beta_{i})$ and $J^{*}_{i}\sim Bernouli(p^{*}_{i})$ independent of $V^{'}_{i}$s.
Here, $N_{1}$ and $N_{2}$ are two discrete random variables with positive integer values and have support as $\{1,2,\ldots,\}$, independent of $T_{i}'$s and ${T^{* }_{i}}'$s, respectively. Under this set-up, here we establish ordering results
between two extreme claim amounts, by using the concept of vector and matrix majorization in the sense of usual stochastic and reversed hazard rate orders. Here, the number of claims $N_{1}$ and $N_{2}$ are random and stochastically 
comparable, independent of $T_{i}$'s and $T^{*}_{i}$'s, respectively. Furthermore, let $T_{n:n}$ and $T_{1:n}$, and $T^{*}_{n:n}$ and $T^{*}_{1:n}$ denote the largest and smallest claim amounts corresponding to the portfolios of risks $\{T_1,\ldots,T_n\}$ and
$\{T^{*}_1,\ldots,T^{*}_n\}$, respectively. We start this section with the result based on usual stochastic order which states that if the matrix of parameters $(\boldsymbol{\psi}(\boldsymbol{p}), \boldsymbol{\alpha};n)$ changes to $(\bm{\psi}(\bm{p}^*),\bm{\beta};n)$ in the sense of the row weak
majorization order in the space $M_{n}$ then $T_{{1}:{N_1}}$ is smaller than $T^{*}_{{1}:{N_2}}$ in terms of usual stochastic order. To establish this theorem, we first prove two comparison results based on vector majorization whereas the first one shows if the modeled parameter are connected with weakly super majorization order then the usual stochastic order holds between two smallest claim amounts for the case of common $\bm{p}$ and the second one provides that the same inequality holds between the smallest claim amounts whenever the occurrence probabilities are associated with weakly super majorization order with common ${\bm\alpha}$.  Throughout this section, we denote $\bm{\psi}(\bm{p})=(\psi(p_1),\ldots,\psi(p_n))=(v_{1},\ldots,v_{n})$ and  $\bm{\psi}(\bm{p^{*}})=(\psi(p^{*}_1),\ldots,\psi(p^{*}_n))=(u_{1},\ldots,u_{n}).$

    \begin{theorem}\label{th_m_st_alpha_small}
	Let $\{U_{1},\ldots,U_{n}\}$ and $\{V_{1},\ldots,V_{n}\}$ be two sets of independent random variables with $U_{i}\sim \bar{F}(x;\alpha_{i})$ and $V_{i}\sim \bar{F}(x;\beta_{i}),$ respectively. Also, let $\{J_{1},\ldots,J_{n}\}$ and 
$\{J_{1}^{*},\ldots,J_{n}^{*}\}$ be another two sets of independent Bernoulli random variables,
independently of $U_{i}'$s and $V_{i}'$s  with
$E(J_{i})=p_{i}$ and $E(J_{i}^{*})=p_{i},$ respectively. Assume that $\bar{F}(x;\alpha_{i})$ is decreasing and log-convex in $\alpha_{i}$ for any $x$.
Then, for $(\boldsymbol{\psi}(\boldsymbol{p}), \boldsymbol\alpha;n), (\boldsymbol{\psi}(\boldsymbol{p}), \boldsymbol{\beta};n)\in M_{n},$ we have $$\bm{\alpha}\overset{w}{\succeq} \bm{\beta} \Rightarrow T_{{1}:{n}}\geq_{st}T^{*}_{{1}:{n}}.$$
	\end{theorem}
    \begin{proof}
        In order to complete the theorem, we only need to show whether  $\bm{\alpha}\overset{w}{\succeq} \bm{\beta} \Rightarrow \bar{F}_{T_{{1}:{n}}}(x)\geq\bar{F}_{T^{*}_{{1}:{n}}}(x),$ which is equivalent to establishing that $\bar{F}_{T_{1:n}}(x)$ is decreasing and Schur-convex in $\bm{\alpha}$ according to Theorem $A.8$ of \cite{Marshall2011}. The reliability function of $\bar{F}_{T_{1:n}}(x)$ can be written as
        $\bar{F}_{T_{1:n}}(x)=\prod_{i=1}^{n}\psi^{-1}(v_{i})\bar{F}(x;\alpha_{i}).$ After taking partial derivative of $\bar{F}_{T_{1:n}}(x)$ with respect to $\alpha_i,$ we have
        \begin{equation}
            \frac{\partial\bar{F}_{T_{1:n}}(x)}{\partial\alpha_{i}}=\frac{\frac{d \bar{F}(x;\alpha_{i})}{d\alpha_{i}}}{\bar{F}(x;\alpha_{i})}\bar{F}_{T_{1:n}}(x)\leq 0,
        \end{equation}
        due to the decreasing property of $\bar{F}(x;\alpha_{i})$ in $\alpha_{i}$.  Consider the case $\alpha_{i}\leq \alpha_{j}$ and $v_{i}\leq v_{j}$ for any pair $i,~j$ such that $1<i\leq j<n.$ Based on the condition that $\bar{F}(x;\alpha_i)$ is log-convex in $\alpha_{i},$ we have the following inequality
        $$(\alpha_{i}-\alpha_{j})\left(\frac{\partial \bar{F}_{T_{1:n}}(x)}{\partial\alpha_{i}}-\frac{\partial \bar{F}_{T_{1:n}}(x)}{\partial\alpha_{j}}\right)\geq 0,$$ which yields $\bar{F}_{T_{1:n}}(x)$ is Schur-convex in $\bm{\alpha}.$ Hence, the proof.
        \end{proof}
        
       \begin{theorem}\label{th_m_st_psi_small}
	Let $\{U_{1},\ldots,U_{n}\}$ and $\{V_{1},\ldots,V_{n}\}$ be two sets of independent random variables with $U_{i}\sim \bar{F}(x;\alpha_{i})$ and $V_{i}\sim \bar{F}(x;\alpha_{i}),$ respectively. Also, let $\{J_{1},\ldots,J_{n}\}$ and 
$\{J_{1}^{*},\ldots,J_{n}^{*}\}$ be another two sets of independent Bernoulli random variables,
independently of $U_{i}'$s and $V_{i}'$s  with
$E(J_{i})=p_{i}$ and $E(J_{i}^{*})=p_{i}^{*},$ respectively. Assume that 
 $\psi:(0,1)\rightarrow (0,\infty)$
is a differentiable, decreasing,  and log-convex
function. 
Then, for $(\boldsymbol{\psi}(\boldsymbol{p}), \boldsymbol\alpha;n), (\boldsymbol{\psi}(\boldsymbol{p}^{*}), \boldsymbol{\alpha};n)\in M_{n},$ we have $$\boldsymbol{\psi}(\boldsymbol{p})\overset {w} {\succeq} \boldsymbol{\psi}(\boldsymbol{p}^{*}) \Rightarrow T_{{1}:{n}}\geq_{st}T^{*}_{{1}:{n}}.$$
	\end{theorem}
 \begin{proof}
        To obtain the required result, we have to prove that $\boldsymbol{\psi}(\boldsymbol{p})\overset {w} {\succeq} \boldsymbol{\psi}(\boldsymbol{p}^{*}) \Rightarrow \bar{F}_{T_{{1}:{n}}}(x)\geq\bar{F}_{T^{*}_{{1}:{n}}}(x).$ In other word, we have to show that $\bar{F}_{T_{1:n}}(x)$ is decreasing and Schur-convex in $\bm{\alpha}$ by Theorem $A.8$ of \cite{Marshall2011}. The reliability function of $\bar{F}_{T_{1:n}}(x)$ can be written as
        $\bar{F}_{T_{1:n}}(x)=\prod_{i=1}^{n}\psi^{-1}(v_{i})\bar{F}(x;\alpha_{i}).$ The partial derivative of $\bar{F}_{T_{1:n}}(x)$ with respect to $v_i,$ is given by
        \begin{equation}
            \frac{\partial\bar{F}_{T_{1:n}}(x)}{\partial v_{i}}=\frac{\frac{d \psi^{-1}(v_{i})}{d v_{i}}}{\psi^{-1}(v_{i})}\bar{F}_{T_{1:n}}(x)\leq 0,
        \end{equation}
        due to the decreasing property of $\psi^{-1}(v_{i})$ in $v_{i}$ for $i=1,\ldots,n$. Consider the case $v_{i}\leq v_{j}$ for any pair $i,~j$ such that $1<i\leq j<n.$ Using the log-convexity property of $\psi^{-1}(v_{i})$ in $v_{i},$ we have
        $$(v_{i}-v_{j})\left(\frac{\partial \bar{F}_{T_{1:n}}(x)}{\partial v_{i}}-\frac{\partial \bar{F}_{T_{1:n}}(x)}{\partial v_{j}}\right)\geq 0,$$ which concludes that $\bar{F}_{T_{1:n}}(x)$ is Schur-convex in $\bm{v}.$ Hence, the proof.
    \end{proof}
    Next, we prove that if $N_{1}\leq_{st} N_{2}$ then the usual stochastic order also holds between two smallest claim amounts $T_{{1}:{N_1}}$ and $T^{*}_{{1}:{N_2}}$ from two sets of independent and heterogeneous portfolios of risks when the matrix of parameters $(\boldsymbol{\psi}(\boldsymbol{p}), \boldsymbol{\alpha};n)$ changes to $(\bm{\psi}(\bm{p}^*),\bm{\beta};n)$ in the sense of the row weak
majorization order in the space $M_{n}.$

    \begin{theorem}\label{th_w_st_small}
	Let $\{U_{1},\ldots,U_{n}\}$ and $\{V_{1},\ldots,V_{n}\}$ be two sets of independent random variables with $U_{i}\sim \bar{F}(x;\alpha_{i})$ and $V_{i}\sim \bar{F}(x;\beta_{i}),$ respectively. Also, let $\{J_{1},\ldots,J_{n}\}$ and 
$\{J_{1}^{*},\ldots,J_{n}^{*}\}$ be another two sets of independent Bernoulli random variables,
independently of $U_{i}'$s and $V_{i}'$s  with
$E(J_{i})=p_{i}$ and $E(J_{i}^{*})=p_{i}^{*},$ respectively. Further, let $N_1$ and $N_2$ be two positive integer-valued random variables independently of $T_{i}'$s and ${T^{*}_{i}}'$s satisfying $N_{1}\leq_{st} N_{2}$, respectively. Assume that the following conditions hold:
\begin{itemize}
    \item [(i)] $\bar{F}(x;\alpha_{i})$ is decreasing and log-convex in $\alpha_{i}$ for any $x$;
        \item [(ii)] $\psi:(0,1)\rightarrow (0,\infty)$
is a differentiable, strictly decreasing and log-convex function.
\end{itemize}
 Then, for $(\boldsymbol{\psi}(\boldsymbol{p}), \boldsymbol\alpha;n), (\boldsymbol{\psi}(\boldsymbol{p}^{*}), \boldsymbol{\beta};n)\in M_{n},$ we have
 $$(\bm{\psi}(\bm{p}),\bm{\alpha};n)\overset{w}{>} (\bm{\psi}(\bm{p}^*),\bm{\beta};n)\Rightarrow T_{{1}:{N_1}}\geq_{st}T^{*}_{{1}:{N_2}}.$$ 
	\end{theorem}
    \begin{proof}
        The main step in proving the theorem is to establish that the following two inequalities hold:
\begin{itemize}
    \item[(i)] $(\bm{\psi}(\bm{p}),\bm{\alpha};n)\overset{w}{>} (\bm{\psi}(\bm{p}^*),\bm{\beta};n)\Rightarrow \bar{F}_{T_{{1}:{n}}}(x)\geq \bar{F}_{T^{*}_{{1}:{n}}}(x);$ 
    \item[(ii)] $\bar{F}_{T_{{1}:{N_{1}}}}(x)\geq \bar{F}_{T^{*}_{{1}:{N_{2}}}}(x).$
\end{itemize}
To complete the first part, let us suppose $W_{1:n},$ $Z_{1:n}$ and $M_{1:n}$ be the smallest claim amounts from sample $J_{p_{1:n}}U_{\alpha_{1:n}},\ldots,J_{p_{n:n}}U_{\alpha_{n:n}},$ 
$J_{p_{1:n}}V_{\beta_{1:n}},\ldots,J_{p_{n:n}}V_{\beta_{n:n}}$ and 
$J^{*}_{p^{*}_{1:n}}V_{\beta_{1:n}},\ldots,J^{*}_{p^{*}_{n:n}}V_{\beta_{n:n}},$ 
respectively. It is easy to observe that $T_{1:n}\overset{st}{=}W_{1:n}$ and $T^{*}_{1:n}\overset{st}{=}M_{1:n}.$ However, according to Theorem \ref{th_m_st_alpha_small} we have $W_{1:n} \geq_{st} Z_{1:n}$ and Theorem \ref{th_m_st_psi_small}, we have $Z_{1:n} \geq_{st} M_{1:n},$ which concludes $\bar{F}_{T_{1:n}}(x)\geq\bar{F}_{T^{*}_{1:n}}(x).$ To prove the second part, first we need to show $\bar{F}_{T^{*}_{1:m}}(x)$ is increasing in $m,$ which can be done by Theorem $1.B.28$ of \cite{shaked2007stochastic}. Now
        \begin{align}
    \bar{F}_{T_{1:{N_{1}}}}(x) 
    &=\sum_{m=1}^{n}P({T_{{1}:{N_{1}}}> x}|N_{1}=m)P(N_{1} =m)\nonumber\\
    &=\sum_{m=1}^{n}P({T_{1:{m}}> x})P(N_{1} =m)\nonumber\\
    &\geq \sum_{m=1}^{n}P({T^{*}_{1:{m}}> x})P(N_{1} =m)~~[\text{as }T_{1:m} \geq_{st} T^{*}_{1:m}]\nonumber\\
    &\geq \sum_{m=1}^{n}P(N_{2} =m)\bar{F}_{T^{*}_{1:{m}}}(x)~~[\text{as } N_{1}\leq_{st} N_{2} \text{ and }\bar{F}_{T^{*}_{1:m}}(x) \text{ is monotone in } m ]\nonumber\\
    &=\bar{F}_{T^{*}_{1:{N_{2}}}}(x),
\end{align}
which completes the proof of the theorem.
    \end{proof}
    The condition ``$N_{1}\leq_{st}N_{2}$" provided in Theorem \ref{th_w_st_small} plays an crucial role to establish the result. In the following counterexample, we see that if the condition ``$N_{1}\leq_{st}N_{2}$" is violated then Theorem \ref{th_w_st_small} does not hold.
    \begin{counterexample}\label{cex3.2}
    Consider the density function of a Gamma distribution (denoted by $\text{Gamma}(\theta,\alpha)$), $f(x,\theta,\alpha)=\frac{1}{\Gamma({\alpha})\theta^{\alpha}}x^{\alpha-1}e^{-\frac{x}{\theta}},~x,~\theta,~\alpha>0.$  Let ${\psi}({p})=-\ln(p)$. Here both $\psi(p)$ and the survival function of Gamma distribution ($\bar{F}(x;\alpha)$), are decreasing and log-convex in $p$ and $\alpha$ for fix $\theta,$ respectively. Suppose, $U_i \sim \text{Gamma}(10.09,\alpha_i)$ and $V_i \sim \text{Gamma}(10.09,\beta_i),$ for $i=1,2,3,4,5.$ Further, set $(\alpha_1,\alpha_2,\alpha_3,\alpha_4,\alpha_5)=(1.9,2,3,5,6)$ and $(\beta_1, \beta_2,\beta_3,\beta_4,\beta_5)=(4.9,6.5,7.6,8.2,10.9)$,  $(p_1,p_2,p_3,p_4,p_5)=(e^{-0.7},e^{-2.1},e^{-3.2},e^{-4.9},e^{-6.9})$ and $(q_1,q_2,q_3,q_4,q_5)=(e^{-1.5},e^{-1.6},e^{-2.6},e^{-3.9},e^{-4.2}).$ Clearly, $(\bm{\psi}(\bm{p}),\bm{\alpha};n)\overset{w}{>} (\bm{\psi}(\bm{p}^*),\bm{\beta};n)$. Let $N_1 \sim \text{Poisson}(\lambda_1)$ and $N_2 \sim \text{Poisson}(\lambda_2),$ where $\lambda_1=10.9$ and $\lambda_2=2.$ Also, consider $(U_{1},U_{2},U_{3})$ are selected with probability $P(N_{1}=3)=\frac{e^{-\lambda_1} \lambda_1^3}{3!},$ $(U_{1},U_{2},U_{3},U_{4})$ are selected with probability $P(N_{1}=4)=\frac{e^{-\lambda_1} \lambda_1^4}{4!}$ and $(U_{1},U_{2},U_{3},U_{4},U_{5})$ are selected with probability $P(N_{1}=5)=\frac{e^{-\lambda_1} \lambda_1^5}{5!},$ and $(V_{1},V_{2},V_{3})$ are selected with probability $P(N_{2}=3)=\frac{e^{-\lambda_2} \lambda_2^3}{3!},$ $(V_{1},V_{2},V_{3}, V_{4})$ are selected with probability $P(N_{2}=4)=\frac{e^{-\lambda_2} \lambda_2^4}{4!}$ and $(V_{1},V_{2},V_{3}, V_{4},V_{5})$ are selected with probability $P(N_{2}=5)=\frac{e^{-\lambda_2} \lambda_2 ^5}{5!}.$  Here all the conditions are satisfied of Theorem \ref{th_w_st_small} except $N_{1}\leq_{st}N_{2}.$ Figure $2(b)$ presents that the difference between $\bar{F}_{T_{1:{N_{1}}}}(x)$ and $\bar{F}_{T^{*}_{{1}:{N_{2}}}}(x)$ crosses $x$ axis for some $x\geq 0$ which violated the statement of the Theorem \ref{th_w_st_small}. 
\end{counterexample}
In \cite{balakrishnan2018}, the authors have proved the following theorem where the usual stochastic order holds between largest claim amounts $T_{n:n}$ and $T^{*}_{n:n}$ when the parameter matrix $(\bm{\psi}(\bm{p}),\bm{\alpha};n)$ changes to another matrix $(\bm{\psi}(\bm{p}^*),\bm{\beta};n).$
\begin{theorem}\label{th_bala}
    Let $\{U_{1},\ldots,U_{n}\}$ and $\{V_{1},\ldots,V_{n}\}$ be two sets of independent random variables with $U_{i}\sim \bar{F}(x;\alpha_{i})$ and $V_{i}\sim \bar{F}(x;\beta_{i}),$ respectively. Also, let $\{J_{1},\ldots,J_{n}\}$ and 
$\{J_{1}^{*},\ldots,J_{n}^{*}\}$ be another two sets of independent Bernoulli random variables,
independently of $U_{i}'$s and $V_{i}'$s  with
$E(J_{i})=p_{i}$ and $E(J_{i}^{*})=p_{i}^{*},$ respectively. 
    Assume that the following conditions hold:
\begin{itemize}
    \item [(i)] $\bar{F}(x;\alpha_{i})$ is decreasing and convex in $\alpha_{i}$ for all $x$;
        \item [(ii)] $\psi:(0,1)\rightarrow (0,\infty)$
is a differentiable, strictly decreasing convex function.
\end{itemize}
 Then, for $(\boldsymbol{\psi}(\boldsymbol{p}), \boldsymbol{\alpha};n), (\boldsymbol{\psi}(\boldsymbol{p}^{*}), \boldsymbol{\beta};n)\in M_{n},$ we have

 $$(\bm{\psi}(\bm{p}),\bm{\alpha};n)\overset{w}{>} (\bm{\psi}(\bm{p}^*),\bm{\beta};n)\Rightarrow T_{{n}:{n}}\geq_{st}T^{*}_{{n}:{n}}.$$ 
\end{theorem}
Now, we generalize Theorem \ref{th_bala} from the case of fixed claim size to the case of random claim sizes.
 
\begin{theorem}\label{th_w_st}
	Let $\{U_{1},\ldots,U_{n}\}$ and $\{V_{1},\ldots,V_{n}\}$ be two sets of independent random variables with $U_{i}\sim \bar{F}(x;\alpha_{i})$ and $V_{i}\sim \bar{F}(x;\beta_{i}),$ respectively. Also, let $\{J_{1},\ldots,J_{n}\}$ and 
$\{J_{1}^{*},\ldots,J_{n}^{*}\}$ be another two sets of independent Bernoulli random variables,
independently of $U_{i}'$s and $V_{i}'$s  with
$E(J_{i})=p_{i}$ and $E(J_{i}^{*})=p_{i}^{*},$ respectively. Further, let $N_1$ and $N_2$ be two positive integer-valued random variables independently of $T_{i}'$s and ${T^{*}_{i}}'$s satisfying $N_{1}\leq_{st} N_{2}$, respectively. Assume that the following conditions hold:
\begin{itemize}
    \item [(i)] $\bar{F}(x;\alpha_{i})$ is decreasing and convex in $\alpha_{i}$ for all $x$;
        \item [(ii)] $\psi:(0,1)\rightarrow (0,\infty)$
is a differentiable, strictly decreasing convex function.
\end{itemize}
 Then, for $(\boldsymbol{\psi}(\boldsymbol{p}), \boldsymbol{\alpha};n), (\boldsymbol{\psi}(\boldsymbol{p}^{*}), \boldsymbol{\beta};n)\in M_{n},$ we have

 $$(\bm{\psi}(\bm{p}),\bm{\alpha};n)\overset{w}{>} (\bm{\psi}(\bm{p}^*),\bm{\beta};n)\Rightarrow T_{{N_1}:{N_1}}\geq_{st}T^{*}_{{N_2}:{N_2}}.$$ 
	\end{theorem}
\begin{proof}
There are two steps to complete the proof of the theorem:
\begin{itemize}
    \item[(i)] $(\bm{\psi}(\bm{p}),\bm{\alpha};n)\overset{w}{>} (\bm{\psi}(\bm{p}^*),\bm{\beta};n)\Rightarrow F_{T_{{n}:{n}}}(x)\leq F_{T^{*}_{{n}:{n}}}(x);$ 
    \item[(ii)] $F_{T_{{N_1}:{N_{1}}}}(x)\leq F_{T^{*}_{{N_2}:{N_{2}}}}(x).$
\end{itemize}
Using Theorem \ref{th_bala}, we can write the first inequality. To establish the second inequality, first we need to show that ${F}_{T^{*}_{m:m}}(x)$ is increasing in $m,$ which can be prove easily as ${F}_{T^{*}_{m:m}}(x)<{F}_{T^{*}_{m+1:m+1}}(x)$ for all $m=1,\ldots,n$. Now
\begin{align}
    F_{T_{{N_{1}}:{N_{1}}}}(x) &=\sum_{m=1}^{n}P({T_{{N_{1}}:{N_{1}}}< x}|N_{1}=m)P(N_{1} =m)\nonumber\\
    &=\sum_{m=1}^{n}P({T_{{m}:{m}}< x})P(N_{1} =m)\nonumber\\
    &\leq\sum_{m=1}^{n}P({T_{{m}:{m}}<x})P(N_{2} =m)~[\textbf{ as } N_{1}\leq_{st}N_{2}] \nonumber\\
    &\leq\sum_{m=1}^{n}P({T^{*}_{{m}:{m}}< x})P(N_{2} =m), ~[\text{ as } T_{{m}:{m}}\geq_{st}T^{*}_{{m}:{m}}]\nonumber\\
    &\leq F_{T^{*}_{{N_{2}}:{N_{2}}}}(x),
\end{align}
 which completes the proof of the theorem.
\end{proof}

Now we consider the following example to illustrate Theorem \ref{th_w_st}.
\begin{example}\label{ex3.1}
    Suppose, $U_i \sim \text{Gamma}(k,\alpha_i)$ and $V_i \sim \text{Gamma}(k,\beta_i),$ with $k=1.5$ for $i=1,2,3,4,5.$ Further, set $(\alpha_1,\alpha_2,\alpha_3,\alpha_4,\alpha_5)=(1.9,2,3,5,6)$ and $(\beta_1, \beta_2,\beta_3,\beta_4,\beta_5)=(4.9,6.5,7.6,8.2,10.9)$, ${\psi}({p})=-\ln(p)$, $(p_1,p_2,p_3,p_4,p_5)=(e^{-0.7},e^{-0.9},e^{-3},e^{-4.9},e^{3.9})$ and $(q_1,q_2,q_3,q_4,q_5)=(e^{-1.2},e^{-1.5},e^{-1.6},e^{-2.6},e^{3.9})$. Clearly, $(\bm{\psi}(\bm{p}),\bm{\alpha};n)\overset{w}{>} (\bm{\psi}(\bm{p}^*),\bm{\beta};n)$. Let $N_1 \sim \text{Poisson}(\lambda_1)$ and $N_2 \sim \text{Poisson}(\lambda_2),$ where $\lambda_1=0.9$ and $\lambda_2=1.9.$ Also, consider $(U_{1},U_{2},U_{3})$ are selected with probability $P(N_{1}=3)=\frac{e^{-\lambda_1} \lambda_1^3}{3!},$ $(U_{1},U_{2},U_{3},U_{4})$ are selected with probability $P(N_{1}=4)=\frac{e^{-\lambda_1} \lambda_1^4}{4!}$ and $(U_{1},U_{2},U_{3},U_{4},U_{5})$ are selected with probability $P(N_{1}=5)=\frac{e^{-\lambda_1} \lambda_1^5}{5!},$ and $(V_{1},V_{2},V_{3})$ are selected with probability $P(N_{2}=3)=\frac{e^{-\lambda_2} \lambda_2^3}{3!},$ $(V_{1},V_{2},V_{3}, V_{4})$ are selected with probability $P(N_{2}=4)=\frac{e^{-\lambda_2} \lambda_2^4}{4!}$ and $(V_{1},V_{2},V_{3}, V_{4},V_{5})$ are selected with probability $P(N_{2}=5)=\frac{e^{-\lambda_2} \lambda_2 ^5}{5!}.$ It is clear that $N_{1}\leq_{st}N_{2}.$ Here, all the conditions are satisfied of Theorem \ref{th_w_st}. Figure $1(a)$ presents that the difference between  ${F}_{T_{N_{1}:{N_{1}}}}(x)$ and ${F}_{T^{*}_{{N_{2}}:{N_{2}}}}(x)$ takes negative values for all $x\geq 0$. 
\end{example}
Next counterexample shows that the result presented in Theorem \ref{th_w_st} may not be true if we remove the condition $(\boldsymbol{\psi}(\boldsymbol{p}), \boldsymbol{\alpha};n), (\boldsymbol{\psi}(\boldsymbol{p}^{*}), \boldsymbol{\beta};n)\in M_{n}.$
\begin{counterexample}\label{cex3.1}
    Suppose, $U_i \sim \text{Gamma}(k,\alpha_i)$ and $V_i \sim \text{Gamma}(k,\beta_i),$ with $k=1.5$ for $i=1,2,3,4,5.$ Further, set $(\alpha_1,\alpha_2,\alpha_3,\alpha_4,\alpha_5)=(1.9,5,5.5,6,10)$ and $(\beta_1, \beta_2,\beta_3,\beta_4,\beta_5)=(0.7,0.9,3,4.9,3.9)$, ${\psi}({p})=-\ln(p)$, $(p_1,p_2,p_3,p_4,p_5)=(e^{-0.7},e^{-0.9},e^{-3},e^{-4.9},e^{3.9})$ and $(q_1,q_2,q_3,q_4,q_5)=(e^{-3.9},e^{-1.5},e^{-1.6},e^{-4.2},e^{2.6})$. Clearly, $(\bm{\psi}(\bm{p}),\bm{\alpha};n)\overset{w}{>} (\bm{\psi}(\bm{p}^*),\bm{\beta};n)$. Let $N_1 \sim \text{Poisson}(\lambda_1)$ and $N_2 \sim \text{Poisson}(\lambda_2),$ where $\lambda_1=0.9$ and $\lambda_2=1.9.$ Also, consider $(U_{1},U_{2},U_{3})$ are selected with probability $P(N_{1}=3)=\frac{e^{-\lambda_1} \lambda_1^3}{3!},$ $(U_{1},U_{2},U_{3},U_{4})$ are selected with probability $P(N_{1}=4)=\frac{e^{-\lambda_1} \lambda_1^4}{4!}$ and $(U_{1},U_{2},U_{3},U_{4},U_{5})$ are selected with probability $P(N_{1}=5)=\frac{e^{-\lambda_1} \lambda_1^5}{5!},$ and $(V_{1},V_{2},V_{3})$ are selected with probability $P(N_{2}=3)=\frac{e^{-\lambda_2} \lambda_2^3}{3!},$ $(V_{1},V_{2},V_{3}, V_{4})$ are selected with probability $P(N_{2}=4)=\frac{e^{-\lambda_2} \lambda_2^4}{4!}$ and $(V_{1},V_{2},V_{3}, V_{4},V_{5})$ are selected with probability $P(N_{2}=5)=\frac{e^{-\lambda_2} \lambda_2 ^5}{5!}.$ Here it is clear that $N_{1}\leq_{st}N_{2}.$  Here all the conditions are satisfied of Theorem \ref{th_w_st} except $(\boldsymbol{\psi}(\boldsymbol{p}), \boldsymbol{\alpha};n), (\boldsymbol{\psi}(\boldsymbol{p}^{*}), \boldsymbol{\beta};n)\in M_{n}.$  Figure $1(b)$ presents that the difference between  ${F}_{T_{N_{1}:{N_{1}}}}(x)$ and ${F}_{T^{*}_{{N_{2}}:{N_{2}}}}(x)$ crosses $x$ axis for some $x\geq 0,$ which violates the statement of the Theorem \ref{th_w_st}. 
    \begin{figure}[ht]
		\begin{center}
  			\subfigure[]{\label{c1}\includegraphics[height=2.0in]{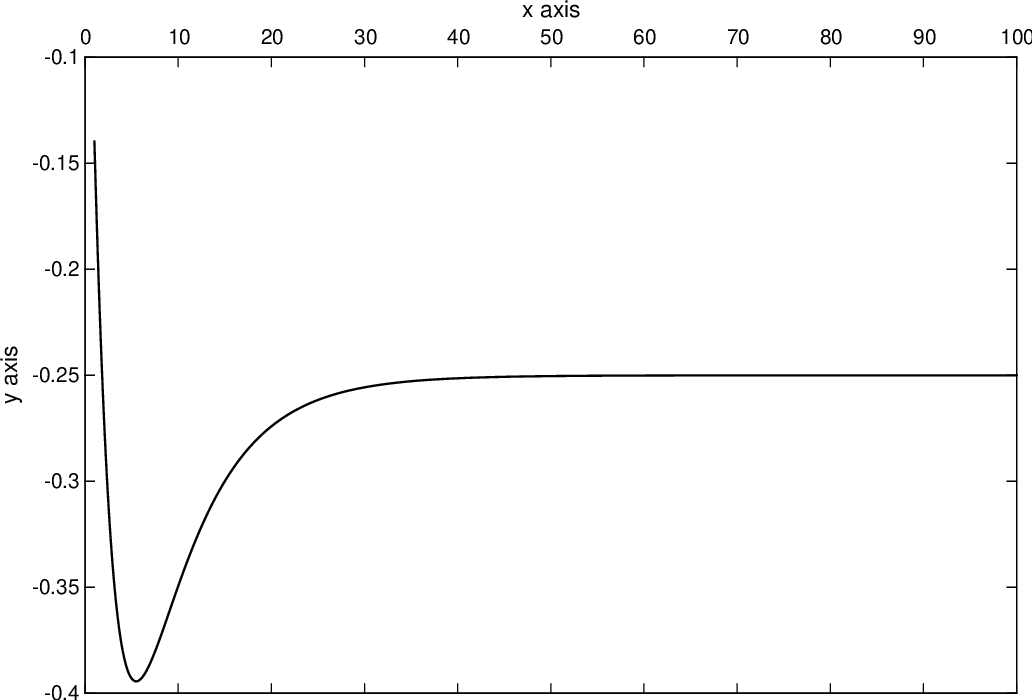}}
			\subfigure[]{\label{c2}\includegraphics[height=2.0in]{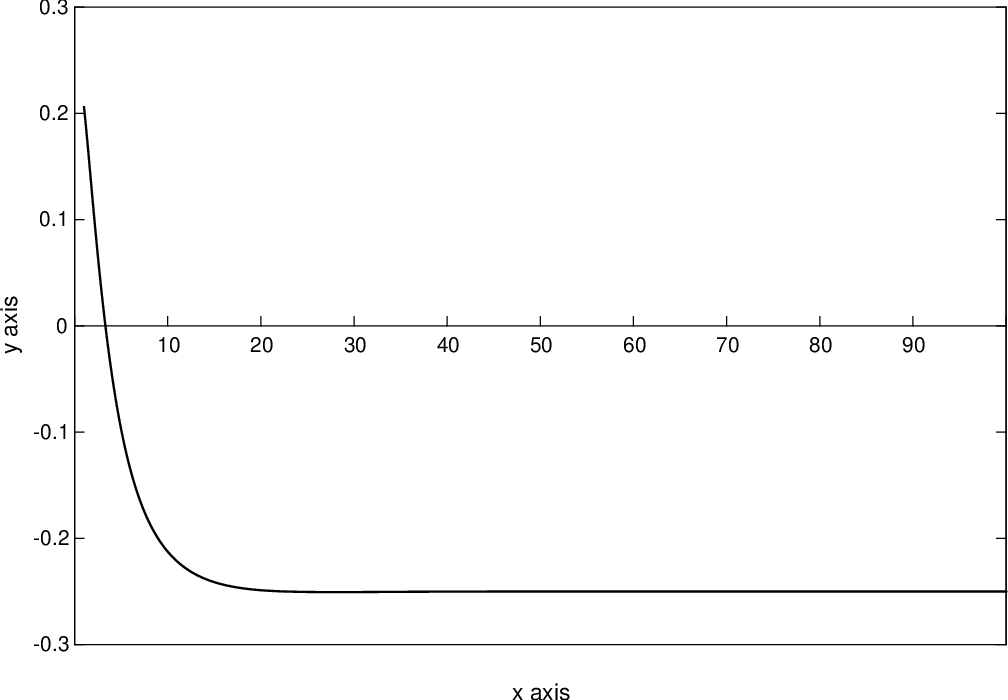}}
			\caption{(a) Plot of $\bar{F}_{T_{{N_1}:{N_{1}}}}(x)-\bar{F}_{T^{*}_{{N_2}:{N_{2}}}}(x)$ as in Example \ref{ex3.1}.  (b) Plot of ${F}_{X_{{N_1}:{N_{1}}}}(x)-{F}_{Y_{{N_2}:{N_{2}}}}(x)$ as in Counterexample \ref{cex3.1}. 
            }
		\end{center}
	\end{figure}
\end{counterexample}
\begin{remark}
   The condition ``$\psi:(0,1)\rightarrow (0,\infty)$
is a differentiable, strictly decreasing convex function" presented in Theorem \ref{th_w_st}, is quite easy to verify for different transform functions. Here, we consider some widely used decreasing and convex transformation functions. Let $X_{i}$ be the $i$th claim amount associated with occurrence probability $p_{i},$ for $i\in\mathcal{J}_{n}.$ Consider
    \begin{itemize}
        \item [(i)] the Exponential transformation with
        $$\psi(p_{i})=e^{-\alpha p_{i}},$$ where $\alpha>0.$
        \item [(ii)] the power (inverse) transformation with
         $$\psi(p_{i})={p^{\beta}_{i}},$$ where $\beta>0$ a risk aversion coefficient, will be used in distorted risk measures, highlighting extreme
tails.
         \item [(iii)] the Negative logarithm with
        $$\psi(p_{i})=-\ln(p_{i}).$$
    \end{itemize}
\end{remark}

    In insurance analysis, it is quite common that the behavior of insurance claims data will be asymmetry and heavy-tailed. Therefore, to model this type data, an insurer always try to find out some distributions having such properties. In this context,  Kumaraswamy-G (Kw-G) distribution is a generalized family of distributions, which is a more flexible model having greater control over the shape, skewness, and kurtosis of the distribution. Due to the flexibility property of Kw-G distributions, insurer
frequently use this distribution to model better tail risks. It is interesting to note that the condition $(i)$ presented in Theorem \ref{th_w_st} contains the Kw-G family, as a special case. Therefore, the following result is a direct application of Theorem \ref{th_w_st} for the case of independent and heterogeneous Kw-G risks. 
The random variable $\{U_{1},\ldots,U_{n}\}$ is said to belong to Kw-G family if $U_{i}\sim (1-G^{\gamma}(x))^{\alpha_{i}},$ where $\gamma,~\alpha_{i}>0$ for $i=1,\ldots,n.$
 \begin{theorem}\label{th_m_st_g}
	Let $\bar{F}(x;\alpha_{i})=(1-G^{\gamma}(x))^{\alpha_{i}}$ and $\bar{F}(x;\beta_{i})=(1-G^{\gamma}(x))^{\beta_{i}}$ for $i=1,\ldots,n.$ Under the set-up of Theorem \ref{th_w_st}, suppose $\psi:(0,1)\rightarrow (0,\infty)$
is a differentiable, strictly decreasing convex function.
Then, for $(\boldsymbol{\psi}(\boldsymbol{p}), \boldsymbol{\alpha};n), (\boldsymbol{\psi}(\boldsymbol{p}^{*}), \boldsymbol{\beta};n)\in M_{n},$ we have

 $$(\bm{\psi}(\bm{p}),\bm{\alpha};n)\overset{w}{>} (\bm{\psi}(\bm{p}^*),\bm{\beta};n)\Rightarrow T_{{N_1}:{N_1}}\geq_{st}T^{*}_{{N_2}:{N_2}},$$ 
	\end{theorem}
  \begin{proof}
     { Using Theorem \ref{th_w_st}, we only need to check that `$\bar{F}(x;\alpha_{i})$ is decreasing and convex in $\alpha_{i}$ for all $x$". For the given model $\bar{F}(x;\alpha_{i})=(1-G^{\gamma}(x))^{\alpha_{i}}$ for $i=1,\ldots,n.$ 
     Taking first order and second order partial derivatives of $\bar{F}(x;\alpha_{i})$ with respect to $\alpha_{i}$ for $i=1,\ldots,n,$ we get
     \begin{equation}\label{e1}
         \frac{\partial\bar{F}(x;\alpha_{i})}{\partial\alpha_i}=log(1-G^{\gamma}(x))(1-G^{\gamma}(x))^{\alpha_{i}}<0
     \end{equation}
     and 
     \begin{equation}\label{e2}
         \frac{\partial^2\bar{F}(x;\alpha_{i})}{\partial\alpha^2_i}=[log(1-G^{\gamma}(x))]^2(1-G^{\gamma}(x))^{\alpha_{i}}>0,
     \end{equation}
 which completes the proof of the theorem. }
  \end{proof} 
    
   Similar to Theorem \ref{th_m_st_g}, it is important to note that the well-known scale family and the proportional hazard rate family also satisfy the condition $(i)$ presented in Theorem \ref{th_w_st}, thus using Theorem \ref{th_w_st}, we can write the following two theorems which are extensions of Theorems $3.5$ and $3.6$ in \cite{balakrishnan2018} from fixed claim size to random claim sizes for the case of heterogeneous and independent scale and proportional hazard rate risks. The random variable $\{U_{1},\ldots,U_{n}\}$ is said to belong to the scale family if $U_{i}\sim \bar{F}(x\alpha_{i}),$ where $\alpha_{i}>0$ for $i=1,\ldots,n.$ Here, $f(x)$ is the density function corresponding to the distribution $F(x).$
  \begin{theorem}\label{th_m_st_scale}
	Let $\bar{F}(x;\alpha_{i})=\bar{F}(x\alpha_{i})$ and $\bar{F}(x;\beta_{i})=\bar{F}(x\beta_{i})$ for $i=1,\ldots,n.$ Under the set-up of Theorem \ref{th_w_st}, suppose the following two conditions hold:
    \begin{itemize}
        \item [(i)] $\psi:(0,1)\rightarrow (0,\infty)$
is a differentiable, strictly decreasing convex function;
\item[(ii)] $f(x)$ is decreasing in $x.$
    \end{itemize}
Then, for $(\boldsymbol{\psi}(\boldsymbol{p}), \boldsymbol\alpha;n), (\boldsymbol{\psi}(\boldsymbol{p}^{*}), \boldsymbol{\beta};n)\in M_{n},$ we have

 $$(\bm{\psi}(\bm{p}),\bm{\alpha};n)\overset{w}{>} (\bm{\psi}(\bm{p}^*),\bm{\beta};n)\Rightarrow T_{{N_1}:{N_1}}\geq_{st}T^{*}_{{N_2}:{N_2}}.$$ 
	\end{theorem}
    \begin{proof}
{ For the given model $\bar{F}(x;\alpha_{i})=\bar{F}(x\alpha_{i})$ for $i=1,\ldots,n.$ 
     Evaluating the first order and second order partial derivatives of $\bar{F}(x;\alpha_{i})$ with respect to $\alpha_{i}$ for $i=1,\ldots,n,$ we obtain
     \begin{equation}\nonumber\label{e.3}
         \frac{\partial\bar{F}(x;\alpha_{i})}{\partial\alpha_i}=-xf(\alpha_i x)<0
     \end{equation}
     and 
     \begin{equation}\nonumber\label{e41}
         \frac{\partial^2\bar{F}(x;\alpha_{i})}{\partial\alpha^2_i}=-x^2\frac{\partial f(\alpha_i x)}{\partial x}>0,
     \end{equation}
     as $f(x)$ is decreasing.
     Therefore, $\bar{F}(x;\alpha_{i})$ is decreasing and convex in $\alpha_{i}$ for all $x$. Finally, using Theorem \ref{th_w_st}, we establish the required result. }
    \end{proof}
In the following, we will apply Theorem \ref{th_w_st} for the proportional hazard rate model. A random variable $U_{i}$ is said to belongs to the proportional hazard rate family if $U_{i}\sim \bar{F}^{\alpha_{i}}(x),$ where $\alpha_{i}>0$ for $i=1,\ldots,n.$
  \begin{theorem}\label{th_m_st_phr}
	Let $\bar{F}(x;\alpha_{i})=\bar{F}^{\alpha_{i}}(x)$ and $\bar{F}(x;\beta_{i})=\bar{F}^{\beta_{i}}(x)$ for $i=1,\ldots,n.$ Under the set-up of Theorem \ref{th_w_st}, suppose $\psi:(0,1)\rightarrow (0,\infty)$
is a differentiable, strictly decreasing convex function.
Then, for $(\boldsymbol{\psi}(\boldsymbol{p}), \boldsymbol\alpha;n), (\boldsymbol{\psi}(\boldsymbol{p}^{*}), \boldsymbol{\beta};n)\in M_{n},$ we have

 $$(\bm{\psi}(\bm{p}),\bm{\alpha};n)\overset{w}{>} (\bm{\psi}(\bm{p}^*),\bm{\beta};n)\Rightarrow T_{{N_1}:{N_1}}\geq_{st}T^{*}_{{N_2}:{N_2}}.$$ 
	\end{theorem}
    \begin{proof}
        { According to Theorem \ref{th_w_st}, it is required to show that $\bar{F}(x;\alpha_{i})$ is decreasing and convex in $\alpha_{i}$ for all $x$. For the given model $\bar{F}(x;\alpha_{i})=\bar{F}^{\alpha_{i}}(x)$ for $i=1,\ldots,n.$ 
     Differentiating $\bar{F}(x;\alpha_{i})$ partially with respect to $\alpha_{i}$ up to second order for $i=1,\ldots,n,$ we get
     \begin{equation}\nonumber\label{e6}
         \frac{\partial\bar{F}(x;\alpha_{i})}{\partial\alpha_i}=\log\bar F(x)\bar{F}^{\alpha_{i}}(x)<0
     \end{equation}
     and 
     \begin{equation}\nonumber\label{e5}
         \frac{\partial^2\bar{F}(x;\alpha_{i})}{\partial\alpha^2_i}=[\log\bar F(x)]^2\bar{F}^{\alpha_{i}}(x)>0,
     \end{equation}
     which completes the proof of the theorem. }
    \end{proof}

    As we know that $A>>B\Rightarrow A>B\Rightarrow A\overset{row}{>}B\Rightarrow A\overset{w}{>}B.$ Thus using Theorem \ref{th_w_st}, we can write the following  result based on multivariate chain majorization order which is a stronger condition. 
  \begin{theorem}\label{th_w_st_multi}
	Let $\{U_{1},\ldots,U_{n}\}$ and $\{V_{1},\ldots,V_{n}\}$ be two sets of independent random variables with $U_{i}\sim \bar{F}(x;\alpha_{i})$ and $V_{i}\sim \bar{F}(x;\beta_{i}),$ respectively. Also, let $\{J_{1},\ldots,J_{n}\}$ and 
$\{J_{1}^{*},\ldots,J_{n}^{*}\}$ be another two sets of independent Bernoulli random variables,
independently of $U_{i}'$s and $V_{i}'$s  with
$E(J_{i})=p_{i}$ and $E(J_{i}^{*})=p_{i}^{*},$ respectively. Further, let $N_1$ and $N_2$ be two positive integer-valued random variables independently of $T_{i}'$s and ${T^{*}_{i}}'$s satisfying $N_{1}\leq_{st} N_{2}$, respectively. Assume that the following conditions hold:
\begin{itemize}
    \item [(i)] $\bar{F}(x;\alpha_{i})$ is decreasing and convex in $\alpha_{i}$ for all $x$;
        \item [(ii)] $\psi:(0,1)\rightarrow (0,\infty)$
is a differentiable, strictly decreasing convex function.
\end{itemize}
 Then, for $(\boldsymbol{\psi}(\boldsymbol{p}), \boldsymbol\alpha;n), (\boldsymbol{\psi}(\boldsymbol{p}^{*}), \boldsymbol{\beta};n)\in M_{n},$ we have

 $$(\bm{\psi}(\bm{p}^*),\bm{\beta};n)=(\bm{\psi}(\bm{p}),\bm{\alpha};n)T \Rightarrow T_{{N_1}:{N_1}}\geq_{st}T^{*}_{{N_2}:{N_2}}.$$  
	\end{theorem}
Similar to Theorem \ref{th_m_st_scale}, using the result presented in Theorem \ref{th_w_st_multi}, we can write the following theorem which is a generalization of Theorem $2$ of \cite{barmalzan2017ordering} for the case of random claim sizes.
      \begin{theorem}\label{th_m_st_scale_multi}
	Let $\bar{F}(x;\alpha_{i})=\bar{F}(x\alpha_{i})$ and $\bar{F}(x;\beta_{i})=\bar{F}(x\beta_{i})$ for $i=1,\ldots,n.$ Under the set-up of Theorem \ref{th_w_st_multi}, suppose the following two conditions hold:
    \begin{itemize}
        \item [(i)] $\psi:(0,1)\rightarrow (0,\infty)$
is a differentiable, strictly decreasing convex function;
\item[(ii)] $f(x)$ is decreasing in $x.$
    \end{itemize}
Then, for $(\boldsymbol{\psi}(\boldsymbol{p}), \boldsymbol\alpha;n), (\boldsymbol{\psi}(\boldsymbol{p}^{*}), \boldsymbol{\beta};n)\in M_{n},$ we have
 $$(\bm{\psi}(\bm{p}^*),\bm{\beta};n)=(\bm{\psi}(\bm{p}),\bm{\alpha};n)T \Rightarrow T_{{N_1}:{N_1}}\geq_{st}T^{*}_{{N_2}:{N_2}}.$$
	\end{theorem}
Still now, our presented comparison results are based on usual stochastic order under the condition that $N_{1}\leq_{st} N_{2}.$ Therefore, it will be a common question whether the usual stochastic order can be extended to other stronger orderings like hazard rate or reversed hazard rate orders. The following counterexample shows that we can not generalized Theorem \ref{th_w_st} from usual stochastic order to reversed hazard rate order under the same conditions.
 \begin{counterexample}\label{ex3.3}
    Under the same set-up as of Example \ref{ex3.1}, Figure $2(a)$ presents that the ratio between  ${F}_{T_{N_{1}:{N_{1}}}}(x)$ and ${F}_{T^{*}_{{N_{2}}:{N_{2}}}}(x)$ is not an increasing function of $x$ for all $x\geq 0,$ which states that Theorem \ref{th_w_st} can not be extended to reversed hazard rate order. 
    \begin{figure}[ht]
		\begin{center}
  			\subfigure[]{\label{c3}\includegraphics[height=2.0in]{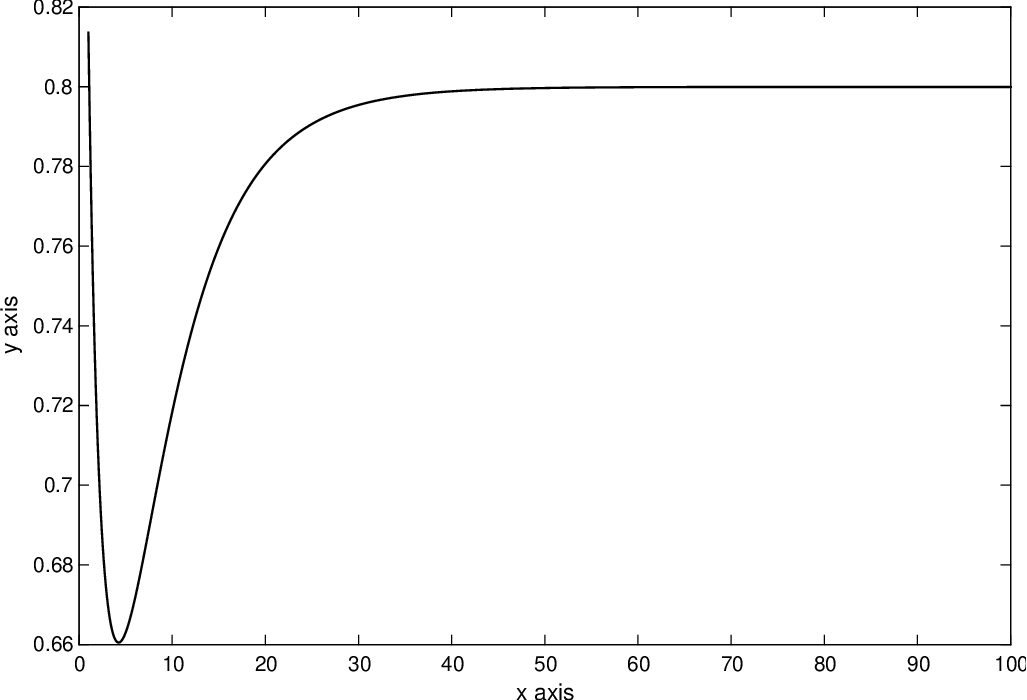}}
			\subfigure[]{\label{c4}\includegraphics[height=2.0in]{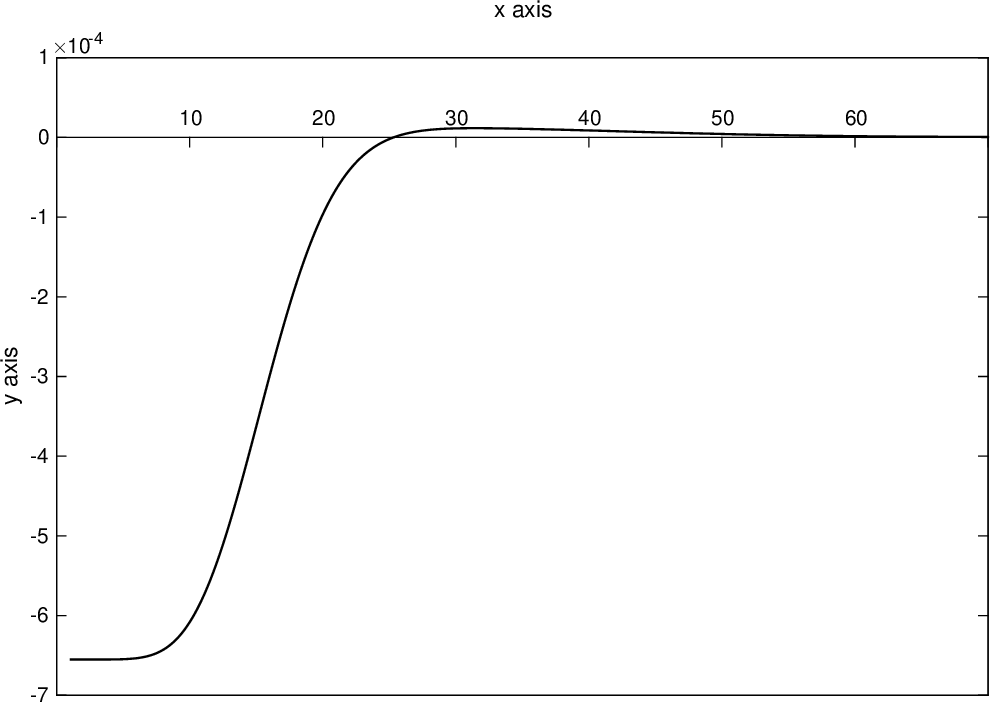}}
			\caption{(a) Plot of ${F}_{T_{{N_1}:{N_{1}}}}(x)/{F}_{T^{*}_{{N_2}:{N_{2}}}}(x)$ as in Counterexample \ref{ex3.3}.  (b) Plot of $\bar{F}_{X_{1:{N_{1}}}}(x)-\bar{F}_{Y_{1:{N_{2}}}}(x)$ as in Counterexample \ref{cex3.2}. 
            }
		\end{center}
	\end{figure}
\end{counterexample}
 
 Next, we establish some certain conditions such that the reversed hazard rate order holds between two largest claim amounts  $T_{{N}:{N}}$ and $T^{*}_{{N}:{N}}$ having random claim size $N$ when the matrix of parameters $(\boldsymbol{\psi}(\boldsymbol{p}), \boldsymbol{\alpha} ;n)$ changes to $(\boldsymbol{\psi}(\boldsymbol{p^*}), \boldsymbol{\beta;}n)$ in the sense of the row weak majorization order in the space $M_{n}$. To complete the theorem, we need to prove the next two results which are related to weakly super majorization order.
\begin{theorem}\label{th_m_rh_alpha}
	Let $\{U_{1},\ldots,U_{n}\}$ and $\{V_{1},\ldots,V_{n}\}$ be two sets of independent random variables with $U_{i}\sim \bar{F}(x;\alpha_{i})$ and $V_{i}\sim \bar{F}(x;\beta_{i}),$ respectively. Also, let $\{J_{1},\ldots,J_{n}\}$ and 
$\{J_{1}^{*},\ldots,J_{n}^{*}\}$ be another two sets of independent Bernoulli random variables,
independently of $U_{i}'$s and $V_{i}'$s  with
$E(J_{i})=p_{i}$ and $E(J_{i}^{*})=p_{i},$ respectively. Assume that the following conditions hold:
\begin{itemize}
    \item [(i)] $\psi:(0,1)\rightarrow (0,\infty)$
is a differentiable and decreasing
function;
 \item [(ii)] $\bar{F}(x;\alpha_{i})$ is decreasing and log-convex in $\alpha_{i}$ for all $x$;
 \item [(iii)] $r(x;\alpha_{i})$ is decreasing and convex in $\alpha_{i}$ for all $x.$
\end{itemize}
Then, for $(\boldsymbol{\psi}(\boldsymbol{p}), \boldsymbol\alpha;n), (\boldsymbol{\psi}(\boldsymbol{p}), \boldsymbol{\beta};n)\in M_{n},$ we have$$\bm{\alpha}\overset{w}{\succeq} \bm{\beta} \Rightarrow T_{{n}:{n}}\geq_{rh}T^{*}_{{n}:{n}}.$$
	\end{theorem}

\begin{proof}
    To complete the result, it is enough to show that $\tilde{r}_{T_{n:n}}(x),$ where
    $\tilde{r}_{T_{n:n}}(x)=\frac{f_{T_{n:n}}(x)}{F_{T_{n:n}}(x)}$ is decreasing and Schur-convex in $\bm{\alpha}$ by using Theorem $A.8$ of \cite{Marshall2011}. 
    Here, $F_{T_{n:n}}(x)$ represents the distribution function of $T_{n:n}$ corresponding to the density function $f_{T_{n:n}}(x)$.
     Then, for $x\geq 0,$
    $$f_{T_{n:n}}(x)={\prod_{i=1}^{n}(1- \psi^{-1}(v_{i})\bar{F}(x;\alpha_{i}))}\sum_{i=1}^{n}\frac{\psi^{-1}(v_{i})\bar{F}(x;\alpha_{i})r(x;\alpha_{i})}{1-\psi^{-1}(v_{i})\bar{F}(x;\alpha_{i})}.$$
Therefore, the reversed hazard rate function of $T_{n:n}$ is given by
\begin{equation}\label{rh}
    \tilde{r}_{T_{n:n}}(x)=\sum_{i=1}^{n}\frac{\psi^{-1}(v_{i})\bar{F}(x;\alpha_{i})r(x;\alpha_{i})}{1-\psi^{-1}(v_{i})\bar{F}(x;\alpha_{i})}.
\end{equation}
Taking partial derivative of $\tilde{r}_{T_{n:n}}(x)$ with respect to $\alpha_{i}$, we get
$$\frac{\partial\tilde{r}_{T_{n:n}}(x)}{\partial \alpha_{i}}=\frac{\frac{d \bar{F}(x;\alpha_{i})}{d \alpha_{i}}}{\bar{F}(x;\alpha_{i})}\frac{r(x;\alpha_{i})\bar{F}(x;\alpha_{i})\psi^{-1}(v_i)}
{(1-\psi^{-1}(v_i)\bar{F}(x;\alpha_{i}))^2}+\frac{dr(x;\alpha_{i})}{d\alpha_{i}}\frac{\bar{F}(x;\alpha_{i})\psi^{-1}(v_i)}
{(1-\psi^{-1}(v_i)\bar{F}(x;\alpha_{i}))^2}\leq 0,$$ 
due to the decreasing properties of $\bar{F}(x;\alpha_{i})$ and $r(x;\alpha_{i})$ in $\alpha_{i}.$
Consider the case $\alpha_{i}\leq \alpha_{j}$ and $v_{i}\leq v_{j}$ for any pair $i,$ $j$ such that $1\leq i< j\leq n.$
According to the conditions that $\bar{F}(x;\alpha_{i})$ is decreasing in $\alpha_{i}$ and $\psi^{-1}(v_{i})$ is decreasing in $v_{i}$, we have
\begin{equation}\label{rh1}
    \frac{r(x;\alpha_{i})\bar{F}(x;\alpha_{i})\psi^{-1}(v_i)}
{(1-\psi^{-1}(v_i)\bar{F}(x;\alpha_{i}))^2}\geq \frac{r(x;\alpha_{j})\bar{F}(x;\alpha_{j})\psi^{-1}(v_j)}
{(1-\psi^{-1}(v_j)\bar{F}(x;\alpha_{j}))^2}.
\end{equation}
Now, the decreasing and log-convexity properties of $\bar{F}(x;\alpha_{i})$ in $\alpha_{i}$ provides
\begin{equation}\label{rh2}
    \frac{\frac{d \bar{F}(x;\alpha_{i})}{d \alpha_{i}}}{\bar{F}(x;\alpha_{i})}\leq \frac{\frac{d \bar{F}(x;\alpha_{j})}{d \alpha_{j}}}{\bar{F}(x;\alpha_{j})}\leq 0.
\end{equation}
As $r(x,\alpha_{i})$ is decreasing and convex in $\alpha_{i},$ thus we have
\begin{equation}\label{rh3}
  \frac{dr(x;\alpha_{j})}{d\alpha_{j}}\frac{\bar{F}(x;\alpha_{j})\psi^{-1}(v_j)}
{(1-\psi^{-1}(v_j)\bar{F}(x;\alpha_{j}))^2}\leq  \frac{dr(x;\alpha_{i})}{d\alpha_{i}}\frac{\bar{F}(x;\alpha_{i})\psi^{-1}(v_i)}
{(1-\psi^{-1}(v_i)\bar{F}(x;\alpha_{i}))^2}\leq 0. 
\end{equation}
Finally, combining equations \eqref{rh1}-\eqref{rh3}, we have the following inequality 
$$(\alpha_{i}-\alpha_{j})\left(\frac{\partial\tilde{r}_{T_{n:n}}(x)}{\partial \alpha_{i}}-\frac{\partial\tilde{r}_{T_{n:n}}(x)}{\partial \alpha_{j}}\right)\geq0,$$
which implies that $\tilde{r}_{T_{n:n}}(x)$ is Schur-convex in $\bm{\alpha}$ using Lemmas \ref{lem2.1}. Hence, the proof.
\end{proof}

\begin{theorem}\label{th_m_rh_psi}
	Let $\{U_{1},\ldots,U_{n}\}$ and $\{V_{1},\ldots,V_{n}\}$ be two sets of independent random variables with $U_{i}\sim \bar{F}(x;\alpha_{i})$ and $V_{i}\sim \bar{F}(x;\alpha_{i}),$ respectively. Also, let $\{J_{1},\ldots,J_{n}\}$ and 
$\{J_{1}^{*},\ldots,J_{n}^{*}\}$ be another two sets of independent Bernoulli random variables,
independently of $U_{i}'$s and $V_{i}'$s  with
$E(J_{i})=p_{i}$ and $E(J_{i}^{*})=p_{i}^{*},$ respectively. Assume that the following conditions hold:
\begin{itemize}
    \item [(i)] $\psi:(0,1)\rightarrow (0,\infty)$
is a differentiable decreasing and log-convex
function; 
 \item [(ii)] $\bar{F}(x;\alpha_{i})$ and $r(x;\alpha_{i}) $ are decreasing in $\alpha_{i}$ for all $x.$
\end{itemize}
Then, for $(\boldsymbol{\psi}(\boldsymbol{p}), \boldsymbol\alpha;n), (\boldsymbol{\psi}(\boldsymbol{p}^{*}), \boldsymbol{\alpha};n)\in M_{n},$ we have $$\bm{\psi} (\bm{p})\overset{w}{\succeq} \bm{\psi} (\bm{p}^{*})\Rightarrow T_{{n}:{n}}\geq_{rh}T^{*}_{{n}:{n}}.$$ 
	\end{theorem}
    \begin{proof}
In order to complete the theorem, it remains to show whether $\tilde{r}_{T_{n:n}}(x)$ is decreasing and Schur-convex in $\bm{v},$ where $\bm{v}=(\psi(p_{1}),\ldots,\psi(p_{n})).$
Taking partial derivative of $\tilde{r}_{T_{n:n}}(x)$ with respect to $v_{i}$, we have 
$$\frac{\partial\tilde{r}_{T_{n:n}}(x)}{\partial v_{i}}=\frac{\frac{d \psi^{-1}(v_i)}{d v_{i}}}{\psi^{-1}(v_i)}\frac{\bar{F}^2(x;\alpha_{i})r(x;\alpha_{i})}{(1-\psi^{-1}(v_i)\bar{F}(x;\alpha_{i}))^2}\leq 0$$ by decreasing property of $\psi^{-1}.$
Consider the case $\alpha_{i}\leq \alpha_{j}$ and $v_{i}\leq v_{j}$ for any pair $i,$ $j$ such that $1\leq i< j\leq n.$
Now, decreasing properties of $\psi^{-1}(v_i)$ in $v_{i},$ $r(x;\alpha_{i})$ in $\alpha_{i}$ and $\bar{F}(x;\alpha_{i})$ in $\alpha_{i}$ provide
\begin{equation}\label{rh4}
    \frac{\bar{F}^2(x;\alpha_{i})r(x;\alpha_{i})}{(1-\psi^{-1}(v_i)\bar{F}(x;\alpha_{i}))^2}\geq \frac{\bar{F}^2(x;\alpha_{j})r(x;\alpha_{j})}{(1-\psi^{-1}(v_j)\bar{F}(x;\alpha_{j}))^2}.
\end{equation}
Using decreasing and log-convexity properties of $\psi^{-1},$ we get 
\begin{equation}\label{rh5}
  \frac{\frac{d \psi^{-1}(v_i)}{d v_{i}}}{\psi^{-1}(v_i)}\leq \frac{\frac{d \psi^{-1}(v_j)}{d v_{j}}}{\psi^{-1}(v_j)}\leq 0.
\end{equation}
Finally, combining equations \eqref{rh4} and \eqref{rh5}, we have the following inequality
\begin{eqnarray}
    &(v_{i}-v_{j})\left(\frac{\partial\tilde{r}_{T_{n:n}}(x)}{\partial v_{i}}-\frac{\partial\tilde{r}_{T_{n:n}}(x)}{\partial v_{j}}\right)
    \geq 0,
\end{eqnarray}
which concludes that $\tilde{r}_{T_{n:n}}(x)$ is decreasing and Schur-convex in $\bm{v}$ according to Theorem $A.8$ of \cite{Marshall2011}. Hence, the proof. 
\end{proof}
\begin{theorem}\label{th_m_rh_lrg}
	Let $\{U_{1},\ldots,U_{n}\}$ and $\{V_{1},\ldots,V_{n}\}$ be two sets of independent random variables with $U_{i}\sim \bar{F}(x;\alpha_{i})$ and $V_{i}\sim \bar{F}(x;\beta_{i}),$ respectively  with $\alpha_{n}\geq \beta _{n}$ and $p_{n}\leq p^{*}_{n}.$ Also, let $\{J_{1},\ldots,J_{n}\}$ and 
$\{J_{1}^{*},\ldots,J_{n}^{*}\}$ be another two sets of independent Bernoulli random variables,
independently of $U_{i}'$s and $V_{i}'$s  with
$E(J_{i})=p_{i}$ and $E(J_{i}^{*})=p_{i}^{*},$ respectively. Further, let $N_1$ and $N_2$ be two positive integer-valued random variables independently of $T_{i}'$s and ${T^{*}_{i}}'$s satisfying $N_{1}\overset{st}{=} N_{2}\overset{st}{=} N$, respectively. Assume that the following conditions hold:
\begin{itemize}
    \item [(i)] $\psi:(0,1)\rightarrow (0,\infty)$
is a differentiable, strictly decreasing and log-convex
function.
\item[(ii)] $\bar{F}(x;\alpha_{i})$ is decreasing and log-convex in $\alpha_{i}$ for all $x;$
\item[(iii)] $r(x;\alpha_{i})$ is decreasing and convex in $\alpha_{i}$ for all $x$.
    \end{itemize}
Then, for $(\boldsymbol{\psi}(\boldsymbol{p}), \boldsymbol\alpha;n), (\boldsymbol{\psi}(\boldsymbol{p}^{*}), \boldsymbol{\beta};n)\in M_{n},$ we have $$(\bm{\psi}(\bm{p}),\bm{\alpha};n)\overset{w}{>} (\bm{\psi}(\bm{p}^*),\bm{\beta};n)\Rightarrow T_{{N}:{N}}\geq_{rh}T^{*}_{{N}:{N}}.$$ 
	\end{theorem}
    \begin{proof}
Applying Theorem $3.4$ of \cite{chowdhury2024}, our aim is to check
\begin{itemize}
    
    \item[(i)] $\frac{{F}_{T_{n:{n}}}(x)}{{F}_{T^{*}_{n:{n}}}(x)}$ is increasing
in $n$;
\item[(ii)] $\tilde{r}_{T^{*}_{n:n}}(x)$ is increasing in $n$; 
    \item[(iii)] $T_{n:n}\geq_{rh} T^{*}_{n:n}.$
\end{itemize}
Denote $A(n)=\frac{{F}_{T_{n:{n}}}(x)}{{F}_{T^{*}_{n:{n}}}(x)}=\frac{\prod_{i=1}^{n}(1-\psi^{-1}(v_{i})\bar{F}(x;\alpha_{i}))}{\prod_{i=1}^{n}(1-\psi^{-1}(u_{i})\bar{F}(x;\beta_{i}))}.$ Using $\alpha_{n}\geq \beta _{n}$ and $p_{n}\leq p^{*}_{n},$
   we can easily verify that $A(n)>A(n-1)$ which implies $\frac{{F}_{T_{n:{n}}}(x)}{{F}_{T^{*}_{n:{n}}}(x)}$ is increasing in $n.$ Also, we can easily check that 
\begin{equation*}
\tilde{r}_{T^{*}_{{n+1}:{n+1}}}(x)-\tilde{r}_{T^{*}_{n:n}}(x)> 0, 
\end{equation*}
which shows $\tilde{r}_{X_{n:n}}(x)$ is increasing in $n.$
    Using the same lines as for Theorem \ref{th_w_st_small}, we can write the following result 
    $$(\bm{\psi}(\bm{p}),\bm{\alpha};n)\overset{w}{>} (\bm{\psi}(\bm{p}^*),\bm{\beta};n)\Rightarrow T_{{n}:{n}}\geq_{rh}T^{*}_{{n}:{n}},$$
    with the help of Theorem \ref{th_m_rh_alpha} and Theorem \ref{th_m_rh_psi}.
       Finally applying Theorem $3.4$ of \cite{chowdhury2024}, we can obtain the required result.
 \end{proof}
It is noteworthy to mention that the condition ``$\bar{F}(x;\alpha_{i})$ is decreasing and log-convex in $\alpha_{i}$ for all $x$" presented in Theorems \ref{th_w_st_small} and \ref{th_m_rh_lrg}, is very important for modeling claims. An insurer usually applies the decreasing survival function to model the time until a claim occurs. In addition, the decreasing log-convexity property can be suitable in a situation where the risk of a claim increases over time, yet the rate of increase gradually slows down. Thus, in insurance analysis, a decreasing log-convex survival function is appropriate as it allows for a more relevant model of risk over time, capturing the evolving behavior of claim probabilities, which will be used for making appropriate premium and capital allocation policies.  In Theorem \ref{th_m_rh_lrg}, the results presented are for the largest claims amounts in the sense of reversed hazard rate order. Now a natural
question that arises whether the result is true for the smallest claim amounts having random claim size. To prove the result, we first need to establish the following two theorems based on vector majorization, where the first theorem shows under some certain conditions if $\bm{\alpha}\overset{m}{\succeq} \bm{\beta}$ then $T_{{1}:{n}}$ is smaller than $T^{*}_{{1}:{n}}$ with respect to the reversed hazard rate order and the second one provides the same ordering holds between $T_{{1}:{n}}$ and $T^{*}_{{1}:{n}}$ whenever $\bm{\psi} (\bm{p})\overset{m}{\succeq} \bm{\psi} (\bm{p}^{*}).$
\begin{theorem}\label{th_m_rh_sm_alpha}
	Let $\{U_{1},\ldots,U_{n}\}$ and $\{V_{1},\ldots,V_{n}\}$ be two sets of independent random variables with $U_{i}\sim \bar{F}(x;\alpha_{i})$ and $V_{i}\sim \bar{F}(x;\beta_{i}),$ respectively. Also, let $\{J_{1},\ldots,J_{n}\}$ and 
$\{J_{1}^{*},\ldots,J_{n}^{*}\}$ be another two sets of independent Bernoulli random variables,
independently of $U_{i}'$s and $V_{i}'$s  with
$E(J_{i})=p_{i}$ and $E(J_{i}^{*})=p_{i},$ respectively. Assume that the following conditions hold:
\begin{itemize}
    \item [(i)] $\psi:(0,1)\rightarrow (0,\infty)$
is a differentiable function; 
\item [(ii)] ${F}(x;\alpha_{i})$ is increasing and log-convex in $\alpha_{i}$ for all $x;$
\item [(iii)] $r(x;\alpha_{i})$ is convex in $\alpha_{i}$ for all $x$.
\end{itemize}
Then, for $(\boldsymbol{\psi}(\boldsymbol{p}), \boldsymbol\alpha;n), (\boldsymbol{\psi}(\boldsymbol{p}), \boldsymbol{\beta};n)\in M_{n},$ we have$$\bm{\alpha}\overset{m}{\succeq} \bm{\beta} \Rightarrow T_{{1}:{n}}\geq_{rh}T^{*}_{{1}:{n}}.$$ 
	\end{theorem}
    \begin{proof}
        In proving the result, utilizing Theorem $A.8$ of \cite{Marshall2011}, we only need to check whether $\tilde{r}_{T_{1:n}}(x)$ is Schur-convex in $\bm{\alpha},$ where
    $\tilde{r}_{T_{1:n}}(x)=\frac{f_{T_{1:n}}(x)}{F_{T_{1:n}}(x)}$ and $\tilde{r}_{T^{*}_{1:n}}(x)=\frac{f_{T^{*}_{1:n}}(x)}{F_{T^{*}_{1:n}}(x)}.$
    Here, $F_{T_{1:n}}(x)$ represents the distribution function of $T_{1:n}$ corresponding to the density function $f_{T_{1:n}}(x)$.
     Then, for $x\geq 0,$
    $$F_{T_{1:n}}(x)=P(T_{1:n}\leq x)=1- \prod_{i=1}^{n}\left(\psi^{-1}(v_{i})\bar{F}(x;\alpha_{i})\right)$$ and $$f_{T_{1:n}}(x)={\prod_{i=1}^{n}\left(\psi^{-1}(v_{i})\bar{F}(x;\alpha_{i})\right)}\sum_{i=1}^{n}r(x;\alpha_{i})={\bar{F}_{T_{1:n}}(x)}\sum_{i=1}^{n}r(x;\alpha_{i}).$$
Therefore, the reversed hazard rate function of $T_{1:n}$ is given by
    $$\tilde{r}_{T_{1:n}}(x)=\frac{\prod_{i=1}^{n}\left(\psi^{-1}(v_{i})\bar{F}(x;\alpha_{i})\right)}{1- \prod_{i=1}^{n}\left(\psi^{-1}(v_{i})\bar{F}(x;\alpha_{i})\right)}\sum_{i=1}^{n}r(x;\alpha_{i})=\frac{\bar{F}_{T_{1:n}}(x)}{F_{T_{1:n}}(x)}\sum_{i=1}^{n}r(x;\alpha_{i}).$$
    Now, taking partial derivative of $\tilde{r}_{T_{1:n}}(x)$ with respect to $\alpha_{i}$ we have
    $$\frac{\partial\tilde{r}_{T_{1:n}}(x)}{\partial\alpha_{i}}=
    \frac{\partial}{\partial\alpha_{i}}\left(\frac{\bar{F}_{T_{1:n}}(x)}{F_{T_{1:n}}(x)}\right)\sum_{i=1}^{n}r(x;\alpha_{i})+\left(\frac{\bar{F}_{T_{1:n}}(x)}{F_{T_{1:n}}(x)}\right)\frac{d r(x;\alpha_{i})}{d\alpha_{i}},$$
    where $\frac{\partial}{\partial\alpha_{i}}\left(\frac{\bar{F}_{T_{1:n}}(x)}{F_{T_{1:n}}(x)}\right)=\frac{\frac{d{F}(x;\alpha_{i})}{d\alpha_{i}}}{{F}(x;\alpha_{i})}\left(\frac{{F}_{T_{1:n}}(x)}{[F_{T_{1:n}}(x)]^2}\right).$
    Now $$(\alpha_{i}-\alpha_{j})\left(\frac{\partial\tilde{r}_{T_{1:n}}(x)}{\partial\alpha_{i}}-\frac{\partial\tilde{r}_{T_{1:n}}(x)}{\partial\alpha_{j}}\right)\geq 0,$$ by the given conditions, which concludes $\tilde{r}_{T_{1:n}}(x)$ is Schur-convex in $\bm{\alpha}.$ Hence, the proof.
    \end{proof}
To prove the next theorem, we need the following result.
    \begin{theorem}\label{ch24_3.3}
    Suppose $X_{1:n}\sim \bar{F}_{1:n}(x)$ and $X_{1:n}\sim \bar{G}_{1:n}(x).$ Also let the support of a positive
integer-valued random variable $N$ having pmf $p(n)$ be $N_{+}$. Now, for all $x\geq l$, if $\tilde{r}_{Y_{1:n}}(x)$,
the reversed hazard rate function of $Y_{1:n}$ is increasing in $n\in N_{+}$ and $\frac{F_{1:n}(x)}{G_{1:n}(x)}$ is increasing
in $n\in N_{+}$, then $X_{1:n}\geq_{rh} Y_{1:n}$ implies $X_{1:N}\geq_{rh} Y_{1:N}$.
\end{theorem}
\begin{proof}
    If $\tilde{r}_{Y_{1:n}}(x)$ is increasing in $n,$ then for any two positive integers $n_{1}\leq n_{2}$ and for all $x\geq l$ we can write $\tilde{r}_{Y_{1:n_{1}}}(x)\leq \tilde{r}_{Y_{1:n_{2}}}(x), $ which is equivalent to say that $\frac{G_{1:n_{1}}(x)}{G_{1}:{n_2}(x)}$ is decreasing in $l\leq x\leq u.$ Thus we can prove the required result using Proposition $3.4$ and Theorem $3.1$ of \cite{Kundu2024}.
\end{proof}

    \begin{theorem}\label{th_m_rh_sm_psi}
	Let $\{U_{1},\ldots,U_{n}\}$ and $\{V_{1},\ldots,V_{n}\}$ be two sets of independent random variables with $U_{i}\sim \bar{F}(x;\alpha_{i})$ and $V_{i}\sim \bar{F}(x;\alpha_{i}),$ respectively. Also, let $\{J_{1},\ldots,J_{n}\}$ and 
$\{J_{1}^{*},\ldots,J_{n}^{*}\}$ be another two sets of independent Bernoulli random variables,
independently of $U_{i}'$s and $V_{i}'$s  with
$E(J_{i})=p_{i}$ and $E(J_{i}^{*})=p_{i}^{*},$ respectively. Assume that $\psi:(0,1)\rightarrow (0,\infty)$
is a differentiable increasing and log-convex
function. Then, for $(\boldsymbol{\psi}(\boldsymbol{p}), \boldsymbol\alpha;n), (\boldsymbol{\psi}(\boldsymbol{p}^{*}), \boldsymbol{\alpha};n)\in M_{n},$ we have$$\bm{\psi} (\bm{p})\overset{m}{\succeq} \bm{\psi} (\bm{p}^{*})\Rightarrow T_{{1}:{n}}\geq_{rh}T^{*}_{{1}:{n}}.$$
	\end{theorem}
    \begin{proof}
    Similar to Theorem \ref{th_m_rh_sm_alpha}, after taking partial derivative of $\tilde{r}_{T_{1:n}}(x)$ with respect to $v_{i},$ we have
    $$\frac{\partial\tilde{r}_{T_{1:n}}(x)}{\partial v_{i}}=
    \frac{\partial}{\partial v_{i}}\left(\frac{\bar{F}_{T_{1:n}}(x)}{F_{T_{1:n}}(x)}\right)\sum_{i=1}^{n}r(x;\alpha_{i}),$$
    where $\frac{\partial}{\partial v_{i}}\left(\frac{\bar{F}_{T_{1:n}}(x)}{F_{T_{1:n}}(x)}\right)=\frac{\frac{d \psi^{-1}(v_{i})}{d v_{i}}}{\psi^{-1}(v_{i})}\left(\frac{\bar{F}_{T_{1:n}}(x)}{[F_{T_{1:n}}(x)]^2}\right).$
    Now $$(v_{i}-v_{j})\left(\frac{\partial\tilde{r}_{T_{1:n}}(x)}{\partial v_{i}}-\frac{\partial\tilde{r}_{T_{1:n}}(x)}{\partial v_{j}}\right)\geq 0,$$ by the given conditions which yields $\tilde{r}_{T_{1:n}}(x)$ is Schur-convex in $\bm{v}.$ 
    Hence, the proof.
    \end{proof}
    \begin{theorem}\label{th_m_rh_sm}
	Let $\{U_{1},\ldots,U_{n}\}$ and $\{V_{1},\ldots,V_{n}\}$ be two sets of independent random variables with $U_{i}\sim \bar{F}(x;\alpha_{i})$ and $V_{i}\sim \bar{F}(x;\beta_{i}),$ respectively. Also, let $\{J_{1},\ldots,J_{n}\}$ and 
$\{J_{1}^{*},\ldots,J_{n}^{*}\}$ be another two sets of independent Bernoulli random variables,
independently of $U_{i}'$s and $V_{i}'$s  with
$E(J_{i})=p_{i}$ and $E(J_{i}^{*})=p_{i}^{*},$ respectively. Further, let $N_1$ and $N_2$ be two positive integer-valued random variables independently of $T_{i}'$s and ${T^{*}_{i}}'$s satisfying $N_{1}\overset{st}{=} N_{2}\overset{st}{=} N$, respectively. Assume that the following conditions hold:
\begin{itemize}
    \item [(i)] $\psi:(0,1)\rightarrow (0,\infty)$
is a differentiable, strictly increasing and log-convex
function;
\item[(ii)] ${F}(x;\alpha_{i})$ is log-convex in $\alpha_{i}$ for all $x;$
\item[(iii)] $r(x;\alpha_{i})$ is convex in $\alpha_{i}$ for any $x$.
        \end{itemize}
Then, for $(\boldsymbol{\psi}(\boldsymbol{p}), \boldsymbol\alpha;n), (\boldsymbol{\psi}(\boldsymbol{p}^{*}), \boldsymbol{\beta};n)\in M_{n},$ we have $$(\bm{\psi}(\bm{p}),\bm{\alpha};n)\overset{row}{>} (\bm{\psi}(\bm{p}^*),\bm{\beta};n)\Rightarrow T_{{1}:{N}}\geq_{rh}T^{*}_{{1}:{N}}.$$
	\end{theorem}
    \begin{proof}
    Applying the similar idea as of Theorem \ref{th_w_st_small}, we can state the following result 
    $$(\bm{\psi}(\bm{p}),\bm{\alpha};n)\overset{row}{>} (\bm{\psi}(\bm{p}^*),\bm{\beta};n)\Rightarrow T_{{1}:{n}}\geq_{rh}T^{*}_{{1}:{n}},$$
    with the help of Theorem \ref{th_m_rh_sm_alpha} and Theorem \ref{th_m_rh_sm_psi}.
       Finally, utilizing Theorem $3.3$ of \cite{chowdhury2024}, one can easily obtain the desired result.
 \end{proof}
 \begin{remark}
     It is quite natural to find some distribution functions having $(i)$ $\bar{F}(x;\alpha_{i})$ is decreasing and log-convex in $\alpha_{i}$ for all $x;$ $(ii)$ $r(x;\alpha_{i})$ is decreasing and convex in $\alpha_{i}$ for all $x;$ $(iii)$ $r(x;\alpha_{i})$ is  convex in $\alpha_{i}$ for all $x$; $(iv)$ ${F}(x;\alpha_{i})$ is log-convex in $\alpha_{i}.$ For example, consider 
     \begin{itemize}
         \item [(i)] exponential distribution with survival function $\bar{F}(x;\alpha)=e^{-\alpha x},~x,\alpha >0,$ which satisfies all the conditions provided in $(i)$-$(iv);$
         \item [(ii)] Weibull distribution with survival function $\bar{F}(x;\alpha,\beta)= e^{-\beta x^{\alpha}},~x, ~\alpha, ~\beta >0,$ which satisfies all the conditions with respect to $\beta$ presented in $(i)$-$(iv);$
         \item [(iii)] Power-generalized Weibull distribution with survival function $\bar{F}(x;c,k)=e^{1-(1+x^c)^{\frac{1}{k}}},~x,~c,~k>0,$ which satisfies properties $(i)$ and $(ii)$ with respect to $k$;
         \item [(iv)] Gamma distribution with density function $f(x;\alpha,\beta)=\frac{1}{\Gamma(\alpha) \beta^{\alpha}}x^{\alpha-1}e^{-\frac{x}{\beta}},~\alpha,~\beta>0,$ which satisfies conditions $(iii)$ and $(iv)$ with respect to $\beta.$
     \end{itemize}
     It is important to note that a distribution having constant hazard rate is very useful in insurance claim modeling, particularly in accident insurance  as the events are random and age independent. { As documented in the literature (see \cite{Mohammed2022} and \cite{Hamza2023}), the exponential distribution, Weibull distribution 
play a crucial role for modeling the claim severity in actuarial science which perform well to analyze skewed data}.
 \end{remark}

\section{Application}\label{app}
   Ordering results between two random extremes is useful in reliability and auction theory. In the following, we consider some applications of our established theoretical results for the purpose of illustration.
\subsection{Reliability theory}
Consider a reliability system having $n$ components in working condition and suppose each components are subjected to a shock that may result in failure. The lifetimes of each components of the system are represented by the non-negative random variables $U_{1},\ldots,U_{n},$ which are experienced with a random shock at the beginning. Let $J_{1},\ldots,J_{n}$ be another collection of independent Bernoulli random variables,
independently of $U_{i}'$s with
$E(J_{i})=p_{i},$ which can be called as shock parameter. Also, let with probability $p_i$ the $i$th component still working after receiving the random shock if $J_{i}=1,$ otherwise it fail to work with probability $1-p_{i}.$ Then the random variable $T_{i}=U_{i}J_{i}$ represent the lifetime of the $i$th component in a system under shock. In reliability analysis, measuring  the probability of remaining healthy or the average lifetime of a system which are under random shock plays a vital role.  Stochastic comparison of such reliability systems has attracted significant research interest. Several comparison results based on fixed sample size for these reliability model have been developed by different authors (see \cite{Fang18}, \cite{kundu2021_shock}, \cite{Abdolahi23} and \cite{Amini24}).  In practice, it is not always possible to get a system having fixed sample size because some observations get
lost for different reasons and sometimes it depends upon the occurrence of some events, which makes the sample
size random. Therefore, analyzing the reliability of a system having random number of components which are under random shock is of great importance from a practical point of view.  In this subsection, we apply our establish results for the above presented model to compare two series as well as parallel systems having random number of components which are under random shock and generalized some existing results from fixed sample size to different (random) sample sizes. 

Let us consider a system having $N_{1}$ number of components, where $U_{1},\ldots,U_{N_{1}}$  with $U_{i}\sim \bar{F}(x;\alpha_{i})$ representing the life-times of the components of the system which may received a random shock at the beginning. Let $\{J_{1},\ldots,J_{N_{1}}\}$ be another set of independent Bernoulli random variables,
independently of $U_{i}'$s with
$E(J_{i})=p_{i}$ such that for a given time period $J_{i}=1$ if the $i$th component is remain operating after an experience of random shocks and, $J_{i}=0$ if the $i$th component fail due to the random shock. Let $T_{i}=U_{i}J_{i},$ for $i=1,\ldots,N_{1}.$ Here, $N_{1}$ is a positive integer-valued random variable with support $\{1,2,\ldots\}$, independently of $T_{i}'$s. Under this set-up, $T_{1},\ldots,T_{N_1}$ represent the life-times of components of the system that are subject to random shocks. Thus, $T_{N_{1}:N_{1}}=\max\{T_{1},\ldots,T_{N_1}\}$ and $T_{1:N_{1}}=\min\{T_{1},\ldots,T_{N_1}\}$ represent the life-times of the parallel and series systems, respectively having random number of observations. Therefore, using our established results presented in Theorem \ref{th_w_st} (Theorem \ref{th_w_st_small}), we can say that under some certain conditions if $N_{1}\leq_{st} N_{2}$ then the survival time of the parallel system $T_{{N_1}:{N_1}}$ (series system $T_{{1}:{N_1}}$) is smaller (greater) than the survival time of the parallel system $T^{*}_{{N_2}:{N_2}}$ (series system $T^{*}_{{1}:{N_2}}$) when the matrix of parameters $(\boldsymbol{\psi}(\boldsymbol{p}), \boldsymbol{\alpha},n)$ changes to $(\bm{\psi}(\bm{p}^*),\bm{\beta},n)$ in the sense of the row weak
majorization order in the space $M_{n}.$ Similarly,
Theorem \ref{th_m_rh_lrg} and Theorem \ref{th_m_rh_sm} provide that based on some certain conditions two parallel as well as series systems that are subject to random shocks are comparable according to the reversed hazard rate order when the matrix of parameters $(\boldsymbol{\psi}(\boldsymbol{p}), \boldsymbol{\alpha},n)$ changes to $(\bm{\psi}(\bm{p}^*),\bm{\beta},n)$ in the sense of the row weak majorization order and row majorize order, respectively in the space $M_{n}.$ Similarly, Theorems \ref{th_m_st_g}, \ref{th_m_st_scale} and \ref{th_m_st_phr} can be interpreted as above when the system components follow Kw-G family, scale family and proportional hazard rate family, respectively.
  \begin{remark}
       It should be mentioned that Theorem \ref{th_m_st_g} is an extension of Theorem $3.3$ of \cite{kundu2021_shock} for the case of random number of shocks. 
    \end{remark}

\subsection{Auction Theory}
 There are various kinds of auctions. Among them, in the first-price reverse auction, the bidders
(sellers) need to provide their sealed bids to the auctioneer (buyer)  who initiates the auction to purchase some specific items. The bidder who submit the lowest price is likely to win the bid and will be received the lowest bid from the auctioneer.
At the same time, due to some unanticipated situations, some of the bidders may wish to opt out from the auction before it starts. For this reason, the auction takes a discussion that the final cost turns out to
be the lowest bid (the smallest order statistics) that comes from $J_{1}U_{1}\ldots,J_{n}U_{n}$, where $J_{i}$ denotes
whether the bidder $i$th participates in the auction or not, and $U_{i}$ is interpreted as the bidding
price for the $i$th bidder if he/she takes a part in the auction, $i=1,\ldots,n$. In the above situation, the number of bidders may be random. 
Therefore, the results presented in the previous section can be applicable in quantitative analysis on
the effects of number of bidders, their attending probabilities and biding price
distributions on the final auction price (see \cite{zhang2019ordering}). 

In this regard, let the bids $U_{1},\ldots,U_{N_{1}}$  with $U_{i}\sim \bar{F}(x;\alpha_{i}).$ Further, let $\{J_{1},\ldots,J_{N_{1}}\}$ be another  set of independent Bernoulli random variables,
independently of $U_{i}'$s with
$E(J_{i})=p_{i}$ such that a given time period $J_{i}=1$ if the $i$th bidder participates in the auction and, $J_{i}=0$ if the $i$th bidder opt out from the auction. Suppose $T_{i}=U_{i}J_{i},$ for $i=1,\ldots,N_{1}.$ Here, $N_{1}$ is the number of bidders which is a  positive integer-valued random variable with support $\{1,2,\ldots\}$, independently of $T_{i}'$s.  Under this set-up,  $T_{1:N_{1}}=\min\{T_{1},\ldots,T_{N_1}\}$ represent the lowest bid amount in the auction where the number of bidder is random. Thus, from Theorem \ref{th_w_st_small}, we can say that under some certain conditions if $N_{1}\leq_{st} N_{2}$ then the final price will be stochastically larger when the matrix of parameters $(\boldsymbol{\psi}(\boldsymbol{p}), \boldsymbol{\alpha},n)$ changes to $(\bm{\psi}(\bm{p}^*),\bm{\beta},n)$ in the sense of the row weak
majorization order in the space $M_{n}.$ Similar interpretation can be found from Theorem \ref{th_m_rh_sm}.
\section{Conclusion}\label{conclusion}
 In this article, we have considered two independent heterogeneous portfolios of risks having different and random number of claims and proved several comparison results based on usual stochastic and reversed hazard rate orders when the matrix of parameters $(\boldsymbol{\psi}(\boldsymbol{p}), \boldsymbol{\alpha},n)$ changes to $(\bm{\psi}(\bm{p}^*),\bm{\beta},n)$ in the sense of the row weak
majorization order and row majorize in the space $M_{n}$.  Furthermore, utilizing our results, we have seen that our obtained results generalize the result established in \cite{barmalzan2017ordering}, \cite{balakrishnan2018} and  \cite{kundu2021_shock}. 

In this work, our developed results are mainly based on independent heterogeneous portfolios of risks for the case of random claim sizes. However, the growing complexity of insurance and reinsurance products has sparked considerable interest in modeling dependent risks. 
In this regard, there are many research available in literature related to dependent portfolios of risks for the case of fixed claim sizes (see \cite{nadeb2020largest, torrado2020, sangita-balakrishnan2022,Liu2022,Panahi2023, Zhang2023}). Thus, we can think to extend the existing results for the case of random claim sizes, which could be some interesting future research problems. It is also important to mention that our established results are mainly based on row weak
majorization order in a particular space. Thus, it will also be interesting to obtain the same results based on other majorization orders (like multivariate chain majorization order) in different spaces. Moreover, our established ordering results are based on the usual stochastic and reversed hazard rate orders. Therefore, we can also extend these results for other stochastic orders. Currently, we are working on it and we hope to report on it in the future. 
\\\\
{\bf Disclosure statement}\\\\
The author states that there is no conflict of interest
\\\\
{\bf Data availability statement}\\\\
Data sharing does not apply to this article as no datasets were generated or analyzed during the current study.
\\\\
{\bf Ethics approval and consent to participate}\\\\
Not applicable.\\\\
{\bf Acknowledgement}\\\\
Sangita Das gratefully acknowledges the financial support for this research work under NPDF, grant No: PDF/2022/000471, ANRF (SERB), Government of India.

\bibliography{SampleReferences}	

@Article{sangita-balakrishnan2022,
  author   = {Sangita Das and Suchandan Kayal and N. Balakrishnan},
  journal  = {Probability in the Engineering and Informational Sciences},
  title    = {Ordering results for smallest claim amounts
	from two portfolios of risks with dependent heterogeneous exponentiated location-scale claims},
  year     = {2022},
  number   = {4},
  pages    = {1116--1137},
  volume   = {36},
  doi      = {10.1017/S0269964821000280},
}

@Article{barmalzan2017ordering,
  author  = {Barmalzan, Ghobad and Payandeh Najafabadi, Amir T. and Balakrishnan, Narayanaswamy},
  journal = {Scandinavian Actuarial Journal},
  title   = {Ordering properties of the smallest and largest claim amounts in a general scale model},
  year    = {2017},
  number  = {2},
  pages   = {105--124},
  volume  = {2017},
  doi     = {10.1080/03461238.2015.1090476},
}

@Article{Das2021mcap,
  author  = {Das, Sangita and Kayal, Suchandan and Balakrishnan, N},
  journal = {Methodology and Computing in Applied Probability},
  title   = {Orderings of the smallest claim amounts from exponentiated location-scale models},
  year    = {2021},
  number  = {3},
  pages   = {971--999},
  volume  = {23},
}

@Book{finkelstein2008failure,
  author    = {Finkelstein, Maxim},
  publisher = {Springer, London},
  title     = {Failure {R}ate {M}odelling for {R}eliability and {R}isk},
  year      = {2008},
  isbn      = {9781848009868},
  series    = {Springer Series in Reliability Engineering},
}

@Book{finkelstein2013stochastic,
  author    = {Finkelstein, Maxim and Cha, Ji Hwan},
  publisher = {Springer, London},
  title     = {Stochastic {M}odeling for {R}eliability, {S}hocks, {B}urn-in and {H}eterogeneous {P}opulations},
  year      = {2013},
  isbn      = {9781447150275},
  series    = {Springer Series in Reliability Engineering},
}

@Book{Marshall2011,
  author    = {Marshall, A. W. and Olkin, I. and Arnold, B. C.},
  publisher = {Springer, New York},
  title     = {Inequalities : {T}heory of {M}ajorization and {I}ts {A}pplications},
  year      = {2011},
  isbn      = {9780387400877},
  series    = {Second edition},
}

@Book{shaked2007stochastic,
  author    = {Shaked, Moshe and Shanthikumar, J George},
  publisher = {Springer, New York},
  title     = {Stochastic {O}rders},
  year      = {2007},
  isbn      = {9780387329154},
}

@Article{marshall1997new,
  author  = {Marshall, Albert W. and Olkin, Ingram},
  journal = {Biometrika},
  title   = {A new method for adding a parameter to a family of distributions with application to the exponential and {W}eibull families},
  year    = {1997},
  number  = {3},
  pages   = {641--652},
  volume  = {84},
  doi     = {10.1093/biomet/84.3.641},
}

@Book{marshall2007life,
  author    = {Marshall, A. W. and Olkin, Ingram},
  publisher = {Springer, New York},
  title     = {Life {D}istributions, {S}tructure of {N}onparametric, {S}emiparametric and {P}arametric {F}amilies},
  year      = {2007},
  isbn      = {9780387203331},
}

@Article{nadeb2020,
  author  = {Nadeb, Hossein and Torabi, Hamzeh and Dolati, Ali},
  journal = {North American Actuarial Journal},
  title   = {Stochastic comparisons between the extreme claim amounts from two heterogeneous portfolios in the case of transmuted-{G} model},
  year    = {2020},
  number  = {3},
  pages   = {475--487},
  volume  = {24},
  doi     = {10.1080/10920277.2019.1671203},
}

@Article{zhang2019ordering,
  author  = {Zhang, Yiying and Cai, Xiong and Zhao, Peng},
  journal = {Astin Bulletin},
  title   = {Ordering properties of extreme claim amounts from heterogeneous portfolios},
  year    = {2019},
  issn    = {0515-0361},
  number  = {2},
  pages   = {525--554},
  volume  = {49},
  doi     = {10.1017/asb.2019.7},
}

@Article{kundu2021_shock,
  author  = {Kundu, Amarjit and Chowdhury,Shovan},
  journal = {Communications in Statistics - Theory and Methods },
  title   = {Ordering properties of the largest order statistics from {K}umaraswamy-{G} models under random shocks},
  year    = {2021},
  number  = {6},
  pages   = {1502--1514},
  volume  = {50},
}

@Article{balakrishnan2018,
  author  = {Balakrishnan, Narayanaswamy and Zhang, Yiying and Zhao, Peng},
  journal = {Scandinavian Actuarial Journal},
  title   = {Ordering the largest claim amounts and ranges from two sets of heterogeneous portfolios},
  year    = {2018},
  number  = {1},
  pages   = {23--41},
  volume  = {2018},
  doi     = {10.1080/03461238.2017.1278717},
}

@Article{nadeb2020largest,
  author  = {Nadeb, Hossein and Torabi, Hamzeh and Dolati, Ali},
  journal = {Mathematical Inequalities \& Applications},
  title   = {Stochastic comparisons of the largest claim amounts from two sets of interdependent heterogeneous portfolios},
  year    = {2020},
  number  = {1},
  pages   = {35--56},
  volume  = {23},
  doi     = {10.7153/mia-2020-23-03},
}

@Article{torrado2020,
  author    = {Nuria Torrado and Jorge Navarro},
  journal   = {Scandinavian Actuarial Journal},
  title     = {Ranking the extreme claim amounts in dependent individual risk models},
  year      = {2021},
  number    = {3},
  pages     = {218--248},
  volume    = {2021},
  doi       = {10.1080/03461238.2020.1830845},
  publisher = {Informa {UK} Limited},
}

@Article{barmalzan2020,
  author  = {Barmalzan, Ghobad and Akrami, Abbas and Balakrishnan, Narayanaswamy},
  journal = {Insurance: Mathematics \& Economics},
  title   = {Stochastic comparisons of the smallest and largest claim amounts with location-scale claim severities},
  year    = {2020},
  pages   = {341--352},
  volume  = {93},
  doi     = {10.1016/j.insmatheco.2020.05.007},
}

@Article{shaked1997,
  author  = {Shaked, Moshe and Wong, Tityik},
  journal = {Journal of Applied Probability},
  title   = {Stochastic comparisons of random minima and maxima},
  year    = {1997},
  issn    = {0021-9002},
  number  = {2},
  pages   = {420--425},
  volume  = {34},
  doi     = {10.2307/3215381},
}

@Article{bartoszewicz2001,
  author  = {Bartoszewicz, J.},
  journal = {Statistics \& Probability Letters},
  title   = {Stochastic comparisons of random minima and maxima from life distributions},
  year    = {2001},
  number  = {1},
  pages   = {107--112},
  volume  = {55},
  doi     = {10.1016/S0167-7152(01)00139-0},
}

@Article{Kundu2024,
  author  = { Pradip Kundu and Amarjit Kundu and Biplab Hawlader},
  journal = {Statistica neerlandica},
  title   = {Stochastic comparisons of largest claim
amounts from heterogeneous portfolios},
  year      = {2023},
  number    = {4},
  pages     = {497--515},
  volume    = {77},
  doi       = {https://doi.org/10.1111/stan.12296},
}

@Article{chowdhury2024,
  author  = { Amarjit Kundu and Shovan Chowdhury and Bidhan Modok},
  journal = {Annals of Operations Research},
  title   = {Stochastic comparisons of random extremes from non-identical random variables},
  year    = {2024},
pages     = {1--16. DOI: 10.1007/s10479-024-06419-1},
}

@Article{Li2004 ,
  author  = {Li, X. and Zuo, M. J.},
  journal = {Naval Research Logistics},
  title   = {Preservation of stochastic orders for random minima and
maxima, with applications},
  pages ={332--344},
  year    = {2004},
  volume  = {51},
  number = {3},
}

@Article{Ahmad2007 ,
  author  = {Ahmad, I. and Kayid, M.},
  journal = {Statistical Papers},
  title   = {Reversed preservation of stochastic orders for random
minima and maxima with applications},
  pages ={283--293},
  year    = {2007},
  volume  = {48},
}

@Article{Nanda2008 ,
  author  = {Nanda, A. K. and Shaked, M.},
  journal = {Communications in Statistics: Theory
Methods},
  title   = {Partial ordering and aging properties of order statistics when sample size is random: a brief review},
  pages ={1710--1720},
  year    = {2008},
  volume  = {37},
}

@Article{Cooner2007,
  author  = {F. Cooner and S. Banerjee and B.P. Carlin and D. Sinha},
  journal = {Journal of the American Statistical
Association},
  title   = {Flexible cure rate modeling under latent activation schemes},
  pages ={560--572},
  year    = {2007},
  volume  = {102},
}

@Article{kumar1994,
  author  = {Dhananjay Kumar and Bengt Klefsjö},
  journal = {Reliability Engineering \& System Safety},
  title   = {Proportional hazards model: a review},
  pages = {177--188},
  year    = {1994},
  volume  = {44},
  number = {2},
}

@Book{Cox1992,
  author  = {Cox, D.R.},
  publisher = {Springer, New York},
  title     = {Regression Models and Life-Tables},
  year      = {1992},
  series    = {In: Kotz, S., Johnson, N.L. (eds) Breakthroughs in Statistics},
}

@Article{Liu2022,
  author  = {Liu, L. and Yan, R.},
  journal = {Journal of Mathematics},
  title   = {Orderings of Extreme Claim Amounts from Heterogeneous and Dependent {W}eibull‐{G } Insurance Portfolios},
  pages = {2768316, 1--21},
  year    = {2022},
  volume  = {2022},
}

@Article{Zhang2023,
  author  = {Jiandong Zhang and Rongfang Yan and  Yiying Zhang},
  journal = {Journal of Computational and Applied Mathematics},
  title   = {Stochastic comparisons of largest claim amount from heterogeneous and dependent insurance portfolios},
  pages = {115265},
  year    = {2023},
  volume  = {431},
}

@Article{Panahi2023,
  author  = {Panahi, A. H. and Jafari, H. and Kia (Barmalzan), G. S.},
  journal = {Communications in Statistics - Theory and Methods},
  title   = {Ordering properties of parallel and series systems with a general lifetime family of distributions for independent components under random shocks},
  pages = {7582--7603},
  year    = {2023},
  volume  = {53},
 number = {21},
}

@Article{Panja2024,
  author  = {Panja, A. and Kundu, P. and Hazra, N. K. and Pradhan, B.},
  journal = { Probability in the Engineering and Informational Sciences},
  title   = {Stochastic comparisons of largest claim and aggregate claim amounts},
  pages = {245--267},
  year    = {2024},
  volume  = {38},
  number = {2},
}

@Article{Hamza2023,
  author  = {Hamza Abubakar and Muhammad Lawal Danrimi},
  journal = {Economic Analysis Letter},
  title   = {A simulation study on the insurance claims distribution using Weibull distribution},
  pages = {67--79},
  year    = {2023},
  volume  = {2},
  number = {3}
}

@Article{Mohammed2022,
  author  = {Mohammed A. Meraou and Noriah M. Al-Kandari and Mohammad Z. Raqab  and Debasis Kundu},
  journal = {Journal of Statistical Computation and Simulation},
  title   = {Analysis of Skewed Data by using Compound Poisson-Exponential Distribution with Applications to Insurance Claims},
  pages = {928--956},
  year    = {2022},
  volume  = {92},
  number = {5}
}

@Article{Laureano2008,
  author  = {Laureano F. Escudero and Eva-María Ortega},
  journal = {Insurance: Mathematics and Economics},
  title   = {Actuarial comparisons for aggregate claims with randomly right-truncated claims},
  pages = {255--262},
  year    = {2008},
  volume  = {43},
  number = {2}
}

@Article{Koutsoyiannis2004,
  author  = {Koutsoyiannis, D.},
  journal = {Hydrological Sciences Journal},
  title   = {Statistics of extremes and estimation of extreme rainfall:1. Theoretical investigation},
  pages = {575–-590},
  year    = {2004},
  volume  = {49},
  number = {4}
}

@Article{Manesh2023,
  author  = {Manesh, SF. and Izadi, M. and Khaledi, B. E.},
  journal = {Probability in the Engineering and Informational Sciences},
  title   = {On stochastic ordering among extreme shock models},
  pages = {961--972},
  year    = {2023},
  volume  = {37},
  number = {4}
}

@Article{Sameen2020,
  author  = {Sameen, Naqvi. and Yiying, Zhang. and  Peng, Zhao},
  journal = {Communications in Statistics - Theory and Methods},
  title   = {Ordering results for individual risk model with dependent Location-Scale claim severities},
  pages = {942--957},
  year    = {2020},
  volume  = {49},
  number = {4}
}

@Article{Sangita_bala_26,
  author  = { Das, S. and  Balakrishnan, N.},
  journal = {Statistics},
  title   = {Ordering results for random maxima and minima from two dependent Kumaraswamy-generalized distributed samples},
  pages = {204--222},
  year    = {2026},
  volume  = {60},
  number = {1}
}

@Article{Fang18,
  author  = { Fang, L. and Balakrishnan, N.},
  journal = {Metrika},
  title   = {Ordering properties of the smallest order statistics from generalized Birnbaum–Saunders models with associated random shocks},
  pages = {19--35},
  year    = {2018},
  volume  = {81}
}

@Article{Amini24,
  author  = { Amini-Seresht, E. and Nasiroleslami, E. and Balakrishnan, N.},
  journal = {Metrika},
  title   = {Comparison of extreme order statistics from two sets of heterogeneous dependent random variables under random shocks},
  pages = {133-153},
  year    = {2024},
  volume  = {87}
}

@Article{Abdolahi23,
  author  = { Abdolahi , M. and Parham , G. and Chinipardaz , R. },
  journal = {REVSTAT-Statistical Journal},
  title   = {Ordering Properties of the Smallest and Largest Order Statistics from Exponentiated Location-Scale Models Under Random Shocks},
  pages = {115-141},
  year    = {2023},
  volume  = {21},
number = {1}
}
	\end{document}